\documentclass[english,12pt]{article}
\usepackage[T1]{fontenc}
\usepackage[latin1]{inputenc}
\usepackage{geometry}
\usepackage{float}
\usepackage{mathrsfs}
\usepackage{amsmath}
\usepackage{amssymb}
\usepackage{graphicx}
\usepackage{tikz}
\usepackage{pgfplots}
\usepackage{bbm}
\usepackage{hyperref}
\hypersetup{
    colorlinks=true,
    linkcolor=blue,
    filecolor=magenta,      
    urlcolor=cyan,
    citecolor = red,
}

\makeatletter

\floatstyle{ruled}
\usepackage{algorithm}
\usepackage{algpseudocode}
\usepackage{caption}

\DeclareCaptionFormat{algor}{%
  \hrulefill\par\offinterlineskip\vskip1pt%
    \textbf{#1#2:} #3\offinterlineskip\hrulefill}
\DeclareCaptionStyle{algori}{singlelinecheck=off,format=algor,labelsep=space}

\usepackage{amsthm}

\usepackage{mathrsfs}

\usepackage{amsfonts}

\usepackage{epsfig}

\usepackage{bm}

\usepackage{mathrsfs}

\usepackage{enumerate}

\@ifundefined{definecolor}{\@ifundefined{definecolor}
 {\@ifundefined{definecolor}
 {\usepackage{color}}{}
}{}
}{}

\usepackage{subcaption}

\newtheorem{theorem}{Theorem}[section]
\newtheorem{lem}{Lemma}[section]
\newtheorem{rem}{Remark}[section]
\newtheorem{prop}{Proposition}[section]

\newtheorem{ass}{Assumption}[section]
\newtheorem{definition}{Definition}[section]

\providecommand{\algorithmname}{Algorithm}
\floatname{algorithm}{\protect\algorithmname}

\usepackage[noabbrev, nameinlink]{cleveref}
\crefname{ass}{assumption}{assumptions}
\crefname{prop}{proposition}{propositions}
\crefname{lem}{lemma}{lemmas}
\crefname{algorithm}{algorithm}{algorithms}
\Crefname{algorithm}{Algorithm}{Algorithms}

\makeatletter
\let\oldtheequation\theequation
\renewcommand\tagform@[1]{\maketag@@@{\ignorespaces#1\unskip\@@italiccorr}}
\renewcommand\theequation{(\oldtheequation)}
\makeatother

\newcounter{hypA}

\newcounter{hypB}

\newcounter{hypD}

\usepackage{babel}\date{}

\usepackage{babel}

\makeatother

\usepackage{babel}

\usepackage{titlesec}
\titleformat{\paragraph}
{\normalfont\normalsize\bfseries}{\theparagraph}{1em}{}
\titlespacing*{\paragraph}
{0pt}{3.25ex plus 1ex minus .2ex}{1.5ex plus .2ex}


\makeatletter
\DeclareRobustCommand\widecheck[1]{{\mathpalette\@widecheck{#1}}}
\def\@widecheck#1#2{%
    \setbox\z@\hbox{\m@th$#1#2$}%
    \setbox\tw@\hbox{\m@th$#1%
       \widehat{%
          \vrule\@width\z@\@height\ht\z@
          \vrule\@height\z@\@width\wd\z@}$}%
    \dp\tw@-\ht\z@
    \@tempdima\ht\z@ \advance\@tempdima2\ht\tw@ \divide\@tempdima\thr@@
    \setbox\tw@\hbox{%
       \raise\@tempdima\hbox{\scalebox{1}[-1]{\lower\@tempdima\box
\tw@}}}%
    {\ooalign{\box\tw@ \cr \box\z@}}}
\makeatother

\begin{document}

\begin{center}

{\Large  \textbf{Unbiased Parameter Inference for a Class of Partially Observed L\'{e}vy-Process Models}}

\vspace{0.4cm}

BY HAMZA RUZAYQAT \&  AJAY JASRA

{\footnotesize Computer, Electrical and Mathematical Sciences and Engineering Division, King Abdullah University of Science and Technology, Thuwal, 23955-6900, KSA.}
{\footnotesize E-Mail:\,} \texttt{\emph{\footnotesize hamza.ruzayqat@kaust.edu.sa, ajay.jasra@kaust.edu.sa}}
\end{center}

\begin{abstract}
We consider the problem of static Bayesian inference for partially observed L\'{e}vy-process models. We develop a methodology which allows one to infer static parameters and some states of the process, 
without a bias from the time-discretization of the afore-mentioned L\'{e}vy process. The unbiased method is exceptionally amenable to parallel implementation and can be computationally efficient relative to competing approaches. We implement the method on S \&P 500 log-return daily data and compare it to some Markov chain Monte Carlo (MCMC) algorithms. 
\end{abstract}

\section{Introduction}
L\'{e}vy processes are one of the most important classes of stochastic processes. They have been widely used in applications including economics, finance, differential geometry, meteorology, turbulence, laser cooling and quantum field theory, to mention just a few. A comprehensive overview of applications of L\'{e}vy processes can be found in 
e.g.~\cite{BMR01, KSW05}. For instance in economics and finance, L\'{e}vy processes are used to model insurance risks, stock prices, exchange and interest rates. They provide more representational models that can capture extreme behaviors like the sudden and discontinuous moves in the price and market volatility, which other classical linear and nonlinear diffusion processes, e.g., Black-Scholes model and Cox-Ingersoll-Ross, usually fail to capture. 

In this article, we are  concerned with statistical (parameter) inference associated to partially observed L\'{e}vy processes that are solutions to stochastic differential equations (SDEs)
and observations are made at discrete and regular times. The particular model structure that we follow is that of  a hidden Markov model (HMM) e.g.~\cite{HMM} with unknown, static (fixed in time) and finite dimensional parameters. In particular, we are concerned with Bayesian static parameter inference associated to one fixed data set.
One of the main challenges with such Bayesian inference is that the transition density of the L\'{e}vy process is often intractable; that is, it does not admit an analytical form, nor does there exist a non-negative and unbiased estimator of it. In such scenarios, one often has to resort to time-discretization, leading to biased inference. In this paper we consider the Euler-discretization for L\'{e}vy driven SDEs, see \cite{R03, PT97}.

Many works in literature have considered Bayesian inference for diffusions and diffusions with jumps. For example, the articles \cite{GW08, RS01} consider fully observed (jump) diffusions, while the works \cite{BS01, GS07, GS06,JKM19, JSDT11} consider jump diffusions that are partially observed at discrete times. The main objective of this article is to perform Bayesian inference of partially observed L\'{e}vy driven SDEs with no time-discretisation error, i.e.~\textit{unbiased Bayesian inference}.  The main benefits for doing so include being able to provide an exact benchmark for inference and an embarrassingly parallel estimator that is based upon independent realizations. Our work is mainly inspired by the unbiased inference approach introduced in \cite{CFJLV21} for partially observed diffusions driven by Brownian motion. In particular, given some function $\varphi$ defined on the space of parameters and hidden states, the objective is to unbiasedly estimate the expectation of $\varphi$ with respect to (w.r.t.) the posterior distribution over the parameters and hidden states, given that it is well-defined, where the latter are diffusions driven by a L\'{e}vy process.

When considering inference that relies on time-discretizations, a high computational cost is often required to achieve smaller bias errors. For instance, as seen in \autoref{sec:numer}, running the particle marginal Metropolis-Hastings (PMMH) algorithm in \cite{ADH10} at a moderate discretization size to infer partially observed diffusions driven by L\'{e}vy process will be computationally expensive compared to the method we propose. The scheme we propose consists of two main tasks; the first is to perform PMMH at a lower discretization level, which can be done cheaply, whereas the second is a correction step. A debiasing scheme as in \cite{RG15} combined with the particle filter (PF) coupling of \cite{JKLZ18} is implemented in the correction step. The unbiasedness scheme is based on similar ideas to multilevel Monte Carlo (MLMC) (see \cite{Giles08}), but instead of optimal allocation of computational resources for minimizing the error, the main focus is to provide unbiased estimators. Ultimately, the approach presented here will result in an unbiased estimator of the expectation of $\varphi$ at a relatively small computational cost, in comparison to some competing methods.

This article is structured as follows. In \autoref{sec:not_model}, we provide the needed notations, the HMM model of interest and the methodology we follow. A brief introduction to L\'{e}vy process is also provided in this section. \autoref{sec:num_levy_and_sde} describes a method to numerically approximate a L\'{e}vy process, and also explains the Euler scheme of \cite{DH11} for SDEs driven by L\'{e}vy processes. We start \autoref{sec:coupled_levy} by reviewing the PF approach and MLMC in the context of SDEs driven by Brownian motions and L\'{e}vy processes. In the same section, we illustrate the process of generating coupled samples of the diffusion process which then used to compute an unbiased estimator of level differences that appear in the debiasing scheme. The overall methodology is detailed in \autoref{sec:ub_infer} where we provide the Bayesian model and main algorithm for unbiased Bayesian inference. Finally, in \autoref{sec:numer}, we implement the algorithm on real daily S \&P 500 log-return data set and compare it to PMMH method.

\section{Partially Observed L\'{e}vy-Process Model}\label{sec:not_model}

In this section we start by giving notations that will be used throughout the paper. Then, we describe the model and methodology adopted and give a brief introduction to L\'{e}vy process.
\subsection{Notations}\label{sec:notat}

Let $(\mathsf{X},\mathcal{X})$ be a measurable space. We denote the empty set by $\phi$. For $\varphi:\mathsf{X}\rightarrow\mathbb{R}$, we write $\mathcal{B}_b(\mathsf{X})$ as the collection of bounded measurable functions. 
Let $\varphi:\mathbb{R}^d\rightarrow\mathbb{R}$, $\textrm{Lip}_{\|\cdot\|_2}(\mathbb{R}^{d})$ denote the collection of real-valued functions that are Lipschitz w.r.t.~$\|\cdot\|_2$ (the $\mathbb{L}_2-$norm of a vector $x\in\mathbb{R}^d$.) That is, $\varphi\in\textrm{Lip}_{\|\cdot\|_2}(\mathbb{R}^{d})$ if there exists a $C<+\infty$ such that for any $(x,y)\in\mathbb{R}^{2d}$,
$
|\varphi(x)-\varphi(y)| \leq C\|x-y\|_2.
$
$\mathcal{P}(\mathsf{X})$ denotes the collection of probability measures on $(\mathsf{X},\mathcal{X})$.
For a measure $\mu$ on $(\mathsf{X},\mathcal{X})$
and a $\varphi\in\mathcal{B}_b(\mathsf{X})$, the notation $\mu(\varphi)=\int_{\mathsf{X}}\varphi(x)\mu(dx)$ is used. 
$\mathsf{B}(\mathbb{R}^d)$ denote the Borel sets on $\mathbb{R}^d$. $dx$ is used to denote the Lebesgue measure.
Let $M:\mathsf{X}\times\mathcal{X}\rightarrow[0,\infty)$ be a non-negative operator and $\mu$ a measure then we use the notations
$
\mu M(dy) = \int_{\mathsf{X}}\mu(dx) M(x,dy)
$
and for $\varphi\in\mathcal{B}_b(\mathsf{X})$, 
$
M(\varphi)(x) = \int_{\mathsf{X}} \varphi(y) M(x,dy).
$
For $A\in\mathcal{X}$ the indicator is written $\mathbb{I}_A(x)$.
$\mathcal{N}_s(\mu,\Sigma)$ denotes an $s-$dimensional Gaussian distribution of mean $\mu$ and covariance $\Sigma$.
For a vector $X$, we use the notation $X_{u:q}$ to denote the vector's elements $\{X_u, X_{u+1},\cdots, X_q\}$. For $X$ \& $Y$ vectors in $\mathbb{R}^d$, $\left\langle X,Y \right\rangle := \sum_{i=1}^d X_i ~Y_i$ is the inner product on $\mathbb{R}^d$. We denote by $I_d$ the identity matrix of size $d\times d$. 
For $\mu\in\mathcal{P}(\mathsf{X})$ and $X$ a random variable on $\mathsf{X}$ with distribution associated to $\mu$ we use the notation $X\sim\mu(\cdot)$. We denote by $a \wedge b$ the minimum of $\{a,b\}$.

\subsection{Model}
\label{subsec:model}
Let $T_f\in\mathbb{N}$ be given and $(\mathsf{X},\mathcal{X},\{\mathcal{X}_t\}_{t\in[0,T_f]},\mathbb{P})$ be a probability space, where $\{\mathcal{X}_t\}_{t\in[0,T_f]}$ is a right-continuous filtration on $\mathcal{X}$ and a stochastic process $\{Y_t\}_{t\in[0,T_f]}$, defined through the following SDE
\begin{align}
\label{eq:levy_driv_sde}
dY_t = f_\theta(Y_{t^-}) ~dX_t,
\end{align}
where $Y_0=y_0\in \mathbb{R}^d$ is known, $\theta\in \Theta \subseteq \mathbb{R}^{d_{\theta}}$ is a static model parameter, $f_\theta:\mathbb{R}^d \to \mathbb{R}^{d\times r}$ is a function, and $\{X_t\}_{t\in[0,T_f]}$ is an $r$-dimensional L\'{e}vy process. For background on L\'{e}vy processes and analysis of SDEs driven by the processes, we refer readers to the books \cite{Applebaum04, BMR01, Bertoin96, KSW05, Protter04, Sato99}. We assume that there are partial noisy observations, $\{Z_i=z_i\}_{i=1}^{T_f}$, $z_i\in\mathbb{R}^m$, of the diffusion process at discrete times $\{1,\dots,T_f\}$; to connect with the standard HMM literature we will set $n=T_f$ and use the notation exchangeably. It is assumed that conditional on $Y_i$, $Z_i$ is independent of the random variables $\{Y_j\}_{j\neq j}$ and that
\begin{align}
\label{eq:obs}
Z_i | Y_i = y_i \sim g_\theta(y_i, \cdot) =: G_i^{(\theta)}(y_i),
\end{align}
where $G_i^{(\theta)}$ is a non-negative, bounded and measurable function that can depend on the unknown parameters $\theta$; we suppress the dependence on the observations. The model given by \ref{eq:levy_driv_sde}-\ref{eq:obs} is an HMM and is often referred to in the literature as the state-space model. To complete the specification, one requires a prior on $\theta$, $\overline{\Pi}(\cdot)$, but for now we shall leave the precise description for later on in the article. We note, however, that we will abuse the notation and use $\overline{\Pi}$ to denote measure and density simultaneously; this applies to all distributions/densities included in this paper as well. Therefore, the model of interest can be written
as
$$
\Pi(d\theta,dy_{1:n}|z_{1:n}) \propto \left(\prod_{i=1}^{n} G_i^{(\theta)}(y_i) M_i^{(\theta,\infty)}(y_{i-1},dy_i)\right)\overline{\Pi}(d\theta)
$$
where $M_i^{(\theta,\infty)}$ is the transition density induced by \eqref{eq:levy_driv_sde} that is assumed to exist; the significance of the superscript $\infty$ will become clear below.
In practice, one can seldom work with the afore-mentioned probability and so one often resorts to a time-discretization of the L\'evy process, in-between observation times.
That is, let $l\in\mathbb{N}\cup\{0\}$ be a parameter that controls the number of discretization points (e.g.~$2^l$), then time discretization induces a collection of auxiliary random
variables $U_i^{l}$ ($i\in\{1,\dots,n\}$) on a space $\mathsf{U}^l$ say, whose dimension grows (in-definitely with $l$), so that one has
\begin{align}
\label{eq:def_M_kernel}
M_i^{(\theta,l)}(y_{i-1},dy_i) := \int_{\mathsf{U}^l} \overline{M}_i^{(\theta,l)}(y_{i-1},dy_i,du_i^l)
\end{align}
for some positive transition kernel $\overline{M}_i^{(\theta,l)}(y_{i-1},dy_i,du_i^l)$. 
Note that $M_i^{(\theta,l)}(y_{i-1},dy_i)$ is constructed so that 
for each $i\in\{1,\dots,n\}$ and $(y_{i-1},A)\in \mathbb{R}^d\times\mathsf{B}(\mathbb{R}^d)$
$$
\lim_{l\rightarrow\infty} \int_A M_i^{(\theta,l)}(y_{i-1},dy_i) = \int_A M_i^{(\theta,\infty)}(y_{i-1},dy_i).
$$
In general, one has to resort to working with
$$
\Pi^{(l)}(d\theta,dy_{1:n},du_{1:n}^l|z_{1:n}) \propto \left(\prod_{i=1}^{n} G_i^{(\theta)}(y_i) \overline{M}_i^{(\theta,l)}(y_{i-1},dy_i,du_i^l)\right)\overline{\Pi}(d\theta).
$$
The exact nature of the random variables $U_i^{l}$ will be discussed later on in the article. We refer to $l$ as the level of discretization.

In this paper, we develop an unbiased static parameter inference for the HMM \ref{eq:levy_driv_sde}-\ref{eq:obs}. The main algorithm consists of two steps:
\begin{enumerate}[(Step1.)]
\item[Step 1.] A PMMH algorithm \cite{ADH10} targeting a coarse-level (low $l$) model that can be implemented inexpensively.
\item[Step 2.] An importance sampling (IS) type correction (see \cite{VHF20}) that comprises
\begin{enumerate}[(a)]
\item a PF coupling as in \cite{JKLZ18},
\item a single-term randomized multilevel Monte Carlo (MLMC) type estimator (see e.g \cite{M11, RG15}), and
\item a multilevel (ML) estimation of the L\'{e}vy driven SDE (see \cite{DH11}).
\end{enumerate} 
\end{enumerate}
\noindent The PMMH algorithm in Step 1 is basically a Metropolis-Hastings (MH) MCMC algorithm accompined with an Euler type discretization of the given SDE. PMMH targets the full joint posterior distribution, $\Pi(d\theta,dy_{1:n}|z_{1:n})$, for certain state-space models by jointly updating the parameter $\theta$ and the states $y_{1:n}$. 

The  posterior $\Pi(d\theta,dy_{1:n}|z_{1:n})$ can be decomposed as 
$$
\Pi(d\theta,dy_{1:n}|z_{1:n}) \propto \overline{\Pi}(d\theta)~\eta_n^{(\theta)}(dy_{1:n}|z_{1:n})~ \mathcal{L}_n^{(\theta)}(z_{1:n}),
$$
where $\eta_n^{(\theta)}(dy_{1:n}|z_{1:n})$ is the joint probability measure of the observations and the diffusion process defined through \ref{eq:levy_driv_sde} at unit times $1,\cdots, n\in \mathbb{N}$, given the parameter $\theta$, and $\mathcal{L}_n^{(\theta)}(z_{1:n})$ is the likelihood of the observations, $Z_{1:n}$, given the parameter $\theta$. The proposal density
$$
q_n(\widetilde{\theta}, \widetilde{y}_{1:n} | \theta, y_{1:n}) := q(\widetilde{\theta}|\theta)~  \eta_n^{(\widetilde{\theta)}}(\widetilde{y}_{1:n}|z_{1:n})
$$
is used in the MH update step to propose $\widetilde{\theta} \sim q(\cdot|\theta)$ and $\widetilde{y}_{1:n} \sim \eta_n^{(\widetilde{\theta})}(\cdot|z_{1:n})$, which then will lead to acceptance probability given by
$$
1 \wedge \frac{ \Pi(\widetilde{\theta},\widetilde{y}_{1:n}|z_{1:n}) ~q_n(\theta, y_{1:n} | \widetilde{\theta}, \widetilde{y}_{1:n})}{\Pi(\theta,y_{1:n}|z_{1:n})~ q_n(\widetilde{\theta}, \widetilde{y}_{1:n} | \theta, y_{1:n}) } = 1 \wedge \frac{ \overline{\Pi}(\widetilde{\theta}) ~\mathcal{L}_n^{(\widetilde{\theta})}(z_{1:n})~q(\theta|\widetilde{\theta})}{\overline{\Pi}(\theta) ~ \mathcal{L}_n^{(\theta)}(z_{1:n})~q(\widetilde{\theta}|\theta)}.
$$
In this article, we use PMMH algorithm to generate samples from the approximate posterior 
$$
\Pi^{(l)}(d\theta,dy_{1:n}|z_{1:n}) = \int_{\mathsf{U}^l} \Pi^{(l)}(d\theta,dy_{1:n},du_{1:n}^l|z_{1:n}),
$$
which requires sampling from an approximation of $\eta_n^{(\theta)}(dy_{1:n}|z_{1:n})$ (namely $\eta_n^{(\theta, l)}$ defined in \autoref{subsec:pf}) that is done through the PF algorithm. The samples resulting from the PF are also used to estimate the likelihood term that appears in the formula of MH acceptance probability.

The bias associated to running PMMH at a low discretization level is dealt with through the correction process in Step 2. The idea is to construct an approximate coupling of the posterior $\Pi$ over the parameter and hidden states at two consecutive discretization levels (with a fixed parameter), namely $l$ and $l-1$, for some $l \in \mathbb{N}$, then correcting by an IS method. This specific coupling will result in an unbiased estimation of the difference of two unnormalised integrals over the hidden states (see \autoref{subsec:cpf_and_ub_level} for more details), a difference that appears in a ML collapsing sum identity. As mentioned in the introduction, the main focus of this article is to unbiasedly estimate the expectation of some given function $\varphi$ that is defined over $\Theta \times \mathbb{R}^{nd}$ w.r.t. the joint posterior distribution $\Pi$ at some time $n\in\mathbb{N}$. The unbiasedness arises from running the coupled PF along with the IS correction at random discretization levels, $l\in \mathbb{N}$, sampled from a given probability mass function (PMF) $P_l$.

This approach was presented in \cite{CFJLV21} for partially observed diffusions driven by a Wiener process. In this article, we apply this method to partially observed models driven by a L\'{e}vy process. In order to properly describe our methodology, we will have to detail several aspects, and this is the topic of the following sections. We start by giving a brief background of L\'{e}vy processes and their properties.
\subsection{L\'{e}vy Processes}
\label{subsec:levy_proc}
L\'{e}vy processes are stochastic processes that are almost surely zero at time zero, continuous in probability with independent and stationary increments. It can be thought of as random walks in continuous time. Examples of L\'{e}vy processes include Wiener, Poisson and Gamma processes. The distribution of a L\'{e}vy process $\{X_t\}_{t\in [0,1]}$ is determined by its characteristic function given by the celebrated L\'{e}vy-Khintchine formula
\begin{align*}
\mathbb{E}[\exp\left\{ i\left\langle w,X_t \right\rangle\right\} ]= \psi_1(t,w) + \psi_2(t,w) + \psi_3(t,w),
\end{align*}
for all $w\in \mathbb{R}^r$, such that
\begin{align*}
\psi_1(t,w) &= \exp\{it\left\langle b, w \right\rangle\}, \quad \psi_2(t,w) = \exp\left\{- \frac{t}{2} \left\langle w, \Sigma w \right\rangle \right\}, \text{ and}\\
\psi_3(t,w)&= \exp\left\{ t \int_{\mathbb{R}^r \setminus{\{0\}}} \left[ \exp\left\{ i\left\langle w,x \right\rangle\right\} -1 - i \left\langle w,x \right\rangle \mathbb{I}_{B_1} (x)\right] \nu(dx)\right\},
\end{align*}
where $b\in \mathbb{R}^r$, $0\leq \Sigma \in \mathbb{R}^{r \times r}$ is symmetric, $B_{\delta}=\{x\in \mathbb{R}^r ~ : ~ \|x\|_2 < \delta\}$, and $\nu$ is called the L\'{e}vy measure of $X_t$ satisfying $\nu(\{0\})=0$ and $\int_{\mathbb{R}^r}(1 \wedge \|x\|_2^2) \nu(dx) < \infty$. Note that this condition on $\nu$ implies that $\nu(B_\delta)<\infty$ for all $\delta>0$, see \cite{Kyp14}. Since probability distributions are uniquely determined by characteristic functions, the L\'{e}vy process $X_t$ is uniquely determined by the so-called L\'{e}vy triplet $(\nu,\Sigma,b)$. The distribution of a L\'{e}vy process is infinitely divisible, that is, the law of a L\'{e}vy process at time $t$ is the sum of the $n$ laws of the L\'{e}vy process increments over time intervals of length $t/n$, where $n$ is some given positive integer. Conversely, if $\mu$ is an infinitely divisible distribution then there exists a L\'{e}vy process $\{X_t\}_{t\in [0,1]}$ such that $X_1\sim \mu$. Moreover,  L\'{e}vy-Khintchine and L\'{e}vy-It\^{o} decomposition theorems propose that every L\'{e}vy process is the sum of three L\'{e}vy processes, which is again a L\'{e}vy process (see \cite{Sato99} for more details.) The terms are a linear Brownian motion (which corresponds to $\psi_1$ and $\psi_2$) and a L\'{e}vy jump process (which corresponds to $\psi_3$.) For a fixed $\delta>0$, note that $\psi_3$ can be written as
\begin{align*}
\psi_3(t,w) &= \exp\Bigg\{t\nu(B_\delta^c) \int_{\mathbb{R}^r} \left[ \exp\left\{ i\left\langle w,x \right\rangle\right\} - 1 \right] \mu(dx) \\
&\qquad\qquad+ t\int_{\mathbb{R}^r}\left[ \exp\left\{ i\left\langle w,x \right\rangle\right\} -1 - i \left\langle w,x \right\rangle \right] \zeta(dx) \Bigg\},
\end{align*}
where
$
\mu(dx)=\mathbb{I}_{B_\delta^c}(x) \nu(dx)/\nu(B_\delta^c)
$
and $\zeta(dx)=\mathbb{I}_{B_\delta\setminus\{0\}}(x) \nu(dx)$. By L\'{e}vy-It\^{o} decomposition theorem, the first integral in the above identity corresponds to a compound Poisson process with intensity $\nu(B_\delta^c)$ and distribution $\mu$, while the second integral corresponds to a compensated Poisson process that is a square-integrable martingale \cite{Kyp14}. Then, any L\'{e}vy process $\{X_t\}_{t\in [0,1]}$ can be written as 
\begin{align}
\label{eq:levy_proc}
X_t = tb + \Sigma^{1/2} W_t + L_t,
\end{align}
where the process $L_t$ is independent of the linear Brownian motion. An important parameter that helps with understanding the behavior of a L\'{e}vy process is Blumenthal-Getoor (BG) index \cite{BG61} defined as
$$
\Gamma := \inf \left\{p > 0~:~ \int_{B_1} |x|^p ~\nu(dx) <\infty\right\} \in [0,2].
$$
The BG index measures the frequency of small jumps. Large index corresponds to a L\'{e}vy process that has small jumps at high frequencies. 

We assume that the L\'{e}vy process $\{X_t\}_{t\in[0,T_f]}$ defined in \ref{eq:levy_proc} and the L\'{e}vy-driven process $\{Y_t\}_{t\in[0,T_f]}$ given in \ref{eq:levy_driv_sde} satisfy the following assumptions:

\begin{ass}
\label{ass:f_and_v}
There exists $C_1,C_2,C_3 > 0$ such that
\begin{itemize}
\item[(i)] $\|f_\theta(x)-f_\theta(y)\|_2 \leq C_1 \|x-y\|_2$ and $\|f_\theta(x)\|_2 \leq C_2$ for all $x\in \mathbb{R}^d$,
\item[(ii)] $0< \int_{\mathbb{R}^d} |x|^2~ \nu(dx) \leq C_3$.
\end{itemize}
\end{ass}
\noindent Note that \autoref{ass:f_and_v} (i) is different from \cite[Assumption A]{DH11} where the authors only assume that $\|f_\theta(y_0)\|_2 \leq C$, with $y_0$ as in \eqref{eq:levy_driv_sde}.

It is generally impossible to exactly simulate the increments of a L\'{e}vy process. Therefore, in addition to an Euler scheme that is used to approximate the process $\{Y_t\}_{t\in[0,T_f]}$, one first has to approximately simulate the L\'{e}vy increments. In this paper, we assume that one cannot exactly sample from the law of $X_t$, hence of $Y_t$, but rather numerical approximation of such processes is adopted. In the following section, we show how one can approximately simulate the increments of a L\'{e}vy process which are then used within an Euler scheme to give approximate samples of the process $\{Y_t\}_{t\geq 0}$.

\section{Numerical Approximations}
\label{sec:num_levy_and_sde}
Here we describe an Euler scheme in which one can produce approximate solutions to the L\'{e}vy-driven SDE in \eqref{eq:levy_driv_sde}. We adopt the approach in \cite{PT97} (see also \cite{JKMP05, R03, DH11}) which is summarized in \autoref{subsec:num_levy_incr} and \autoref{subsec:Euler_sde} below.

\subsection{A Numerical Approximation to L\'{e}vy Process}
\label{subsec:num_levy_incr}
We first consider an approximation scheme to simulate the L\'{e}vy process over the time interval $[0,1]$. Let $\Delta_l = 2^{-l}$, $l\in \mathbb{N}\cup \{0\}$, be given and denote by $\delta_l>0$ the jump threshold parameter so that jumps with height smaller than $\delta_l$ are ignored. We will refer to $l$ as the level of discretization. Define $\lambda_l := \nu(B_{\delta_l}^c)$, then given $l\in \mathbb{N}\cup \{0\}$, we choose $\delta_l$ such that $\lambda_l = \Delta_l^{-1}$, which assures that $\lambda_l > 0$ for sufficiently small $\delta_l>0$. Let $\{N(t,A)\}_{t\geq 0, A\in \mathsf{B}(\mathbb{R}^r\setminus\{0\})}$ denote a Poisson process with rate $\nu(A)$. Then, the number of jumps with heights larger than $\delta_l$ before time $t$ is a Poisson process $N(t, B_{\delta_l}^c)$ with rate $\lambda_l$. We denote by $\Delta L_t := L_t - \lim_{s \uparrow t} L_s$ the jump height at discontinuity time $t$ and $\Delta L_{\widehat{T}_i^l} :=L_{\widehat{T}_i^l} - L_{\widehat{T}_{i-1}^l}$ the jump height at discontinuity time $\widehat{T}_i^l$, where $l$ is the discretization level and $i \in \{1,\cdots , \widehat{K} \}$ for some $\widehat{K}\in \mathbb{N}$ such that $\widehat{T}_0^l=0$ and $\widehat{T}_{\widehat{K}}^l = 1$. Set 
$$
F_0^l := \int_{B_{\delta_l}^c} x~ \nu(dx)
$$
and
$$
\mu^l(dx) := \frac{1}{\lambda_l}\mathbb{I}_{B_{\delta_l}^c}(x) \nu(dx),
$$
then $\mu^l(dx)$ defines a probability measure on $\mathbb{R}^r\setminus\{0\}$ such that $\int_{B_{\delta_l}^c} x~ N(t,dx)$ is a random finite sum of random variables that are sampled from $\mu^l(dx)$. These random variables are the heights of all jumps larger than $\delta_l$ before time $t$. The compensated compound Poisson process defined by 
$$
L^{\delta_l}_t :=\int_{B_{\delta_l}^c} x~ N(t,dx) - t F_0^l =  \sum_{i = 1}^{N(t,B_{\delta_l}^c)} \Delta L_{\widehat{T}_i^l} - t F_0^l
$$
is an $L_2$ martingale which converges in $L_2$ as $\delta_l\to 0$ to the L\'{e}vy process $L_t$ in \ref{eq:levy_proc} \cite{Applebaum04, DH11}. The term $tF_0^l$ is the expected total length of all jumps before time $t$ given by $\mathbb{E}\Big[ \int_{B_{\delta_l}^c} x~ N(t,dx) \Big]$. The time increments $\{\widehat{T}_i^l - \widehat{T}_{i-1}^l\}_{i=1}^{\widehat{K}}$ of the jumps with heights larger than $\delta_l$ are sampled from an exponential distribution with parameter $\lambda_l$. However, since we want the Brownian increments $W_{\widehat{T}_i^l}-W_{\widehat{T}_{i-1}^l}$ to be sampled from a normal distribution with covariance matrix $\alpha_i\,I_r$ such that $\alpha_i$ is at most $\Delta_l$, a refinement of the jump times is required. Set $T_0^l=0$, then the new jump times are
$$
T_i^l = \min \Big\{T_{i-1}^l+\Delta_l,~ \min\{\widehat{T}_j^l>T_{i-1}^l~:~ j\in\{1,\cdots,\widehat{K}\} \Big\},
$$
for $i \in \{1,\cdots,K\}$, where $K\in \mathbb{N}$ and $T_{K}^l=1$. To summarize, increments of the process $\{X_t\}_{t\in[0,1]}$ are approximated at discretization times $\{T_i^l\}_{i=1}^{K}$  by 
\begin{align}
\label{eq:disc_levy}
\Delta X_{T_i^l}^l = (b-F_0^l) (T_i^l - T_{i-1}^l) + \Sigma^{1/2}(W_{T_i^l}-W_{T_{i-1}^l}) + \Delta L_{T_i^l}.
\end{align}
If $T_i^l = \widehat{T}_j^l$ for some $j\in\{1,\cdots,\widehat{K}\}$, then $\Delta L_{T_i^l} =  \Delta L_{\widehat{T}_j^l} \sim \mu^l$, otherwise $\Delta L_{T_i^l} = 0$. One can use the inversion sampling method to sample from $\mu^l$ if the cumulative distribution function can be computed analytically or else other methods can be used (e.g. rejection sampling.) Below, we state a needed assumption for the choice of $\delta_l$.
\begin{ass}
\label{ass:delta_l}
There exists a $\kappa>0$ such that for all $l\in\mathbb{N}$, there exists a $\delta_l>0$ such that $\delta_l=\mathcal{O}(\Delta_l^\kappa)$ and $\lambda_l = \Delta_l^{-1}$.
\end{ass}
In \autoref{alg:disc_levy} we summarize the process of obtaining a single-level approximation of increments of the L\'{e}vy process on the interval $[0,1]$.
%
%
\begin{center}
\captionsetup[algorithm]{style=algori}
\captionof{algorithm}{Single-level increments of the L\'{e}vy process on $[0,1]$}
\label{alg:disc_levy}
\raggedright
\begin{enumerate}
\item \textbf{Input:} $l$, $\lambda_l$ and $\mu^l(dx)$.
\item Initialization: Set $\widehat{T}_0^l = 0$, $T_0^l=0$ and $i=1$.
\item Jump times:
\begin{itemize}
\item Set $\widehat{T}_i^l = \min \{1, \widehat{T}_{i-1}^l + \zeta_i \}$, where $\zeta_i \sim \text{Exp}(\lambda_l)$.
\item If $\widehat{T}_i^l = 1$, set $\widehat{K}=i$ and go to Step 4..
\item Otherwise, $i = i + 1$ and go to start of Step 3..
\end{itemize}
\item Jump heights: For $i = 1, \cdots, \widehat{K}$, sample $\Delta L_{\widehat{T}_i^l} \sim \mu^l$. Set $i = 1$.
\item Refinement of time increments: 
\begin{itemize}
\item $T_i^l = \min \Bigg\{T_{i-1}^l+\Delta_l,~ \min \Big\{\widehat{T}_j^l>T_{i-1}^l~:~ j\in\{1,\cdots,\widehat{K}\} \Big\} \Bigg\}$.
\item If $T_i^l = \widehat{T}_j^l$ for some $j\in \{1,\cdots,\widehat{K}\}$, set $\Delta L_{T_i^l} =  \Delta L_{\widehat{T}_j^l} $; otherwise $\Delta L_{T_i^l} =  0$.
\item If $T_i^l = 1$, set $K = i$ and go to output Step 6.
\item Otherwise, $i = i + 1$ and go to start of Step 5..
\end{itemize} 
\item \textbf{Output}: Return $K$, $\{T_i^l\}_{i=0}^K$ and $\{\Delta L_{T_i^l}\}_{i=1}^K$.
\end{enumerate}
\vspace{-0.2cm}
\hrulefill
\end{center}
\subsection{An Euler Scheme to L\'{e}vy-driven SDEs}
\label{subsec:Euler_sde}
Assume that $\theta\in \Theta$ is fixed. The Euler approximation of the L\'{e}vy-driven process $\{Y_t\}_{t\in[0,1]}$ at discretization times $\{T_i^l\}_{i=1}^{K}$ is given by
\begin{align}
\label{eq:disc_levy_driv}
Y_{T_i^l}^l = Y_{T_{i-1}^l}^l + f_\theta(Y_{T_{i-1}^l}^l) \Delta X_{T_i^l}^l,
\end{align}
with $Y_{T_0^l}^l = Y_0$ and the L\'{e}vy increments $\{\Delta X_{T_i^l}^l\}_{i=1}^{K}$ are computed through \eqref{eq:disc_levy}. \autoref{alg:disc_levy_driv} describes the process of generating samples $Y_1^l$, an Euler approximation of the process $\{Y_t\}_{t\geq 0}$ at time 1 and a discretization level $l$. In this paper, we shall consider the discrete sequence $\{Y_0, Y_1, \cdots, Y_{T_f}\}$, where for each $n=1,\cdots,T_f$ we approximate $Y_n$ by $Y_n^l$, where $Y_n^l=Y_{T_K^l}^l$ and in the input of \autoref{alg:disc_levy_driv} $Y_{0}^l$ is set to $Y_{n-1}^l$.

Theorem 2 of \cite{DH11} provides asymptotic bounds for the strong error of the Euler scheme presented in \autoref{alg:disc_levy_driv}. The theorem states that under \autoref{ass:f_and_v}, for any $l \in \mathbb{N}$ and $\delta_l>0$ such that $\nu(B_{\delta_l}^c)=\Delta_l^{-1}$, there exists a $C<\infty$ such that
\begin{align}
&\mathbb{E} \left[\sup_{t\in [0,1]} \left|Y_t - Y_t^l \right|^2 \right] \leq C (\sigma_{\delta_l}^2 + \Delta_l |\log(\Delta_l)|), \qquad \Sigma \neq 0, \nonumber\\
&\mathbb{E} \left[ \sup_{t\in [0,1]} \left|Y_t - Y_t^l \right|^2 \right] \leq C (\sigma_{\delta_l}^2 + |b- F_0^l|^2 \Delta_l^2), \qquad \Sigma = 0, \label{eq:strong_err_Euler}
\end{align}
where $\sigma_{\delta_l}^2 := \int_{B_{\delta_l}} |x|^2~\nu(dx)$. As we can observe from \eqref{eq:strong_err_Euler} the strong error rate of the Euler scheme will only depend upon the choice of $\nu$ and whether the L\'{e}vy process has a Brownian motion component. We will show in \autoref{sec:numer} that for a certain choice of $\nu$, the strong error rate $\beta$ is 3.
%
%
\begin{center}
\captionsetup[algorithm]{style=algori}
\captionof{algorithm}{Single-level sampling for L\'{e}vy-driven process}
\label{alg:disc_levy_driv}
\raggedright
\begin{enumerate}
\item \textbf{Input:} $Y_{0}$, $l$, $\lambda_l$, $F_0^l$, $\mu^l(dx)$ and $\theta$.
\item Call \autoref{alg:disc_levy} with input $l$, $\lambda_l$ and $\mu^l(dx)$ to return $\{T_i^l\}_{i=0}^K$ and $\{\Delta L_{T_i^l}\}_{i=1}^K$.
\item For $i = 1, \cdots,K$, sample the Brownian increments $W_{T_i^l} - W_{T_{i-1}^l}$ from $\mathcal{N}_r(0,(T_i^l - T_{i-1}^l)I_r)$. 
\item With $\{T_i^l\}_{i=0}^K$, $\{\Delta L_{T_i^l}\}_{i=1}^K$, $\{W_{T_i^l} - W_{T_{i-1}^l}\}_{i=1}^K$ and $F_0^l$, generate the increments $\{\Delta X_{T_i^l}\}_{i=1}^K$ through \ref{eq:disc_levy}. Then, compute  $Y_1^l$ via the recursion in \ref{eq:disc_levy_driv}.
\item \textbf{Output:} Return $Y_1^l$.

\end{enumerate}
\vspace{-0.2cm}
\hrulefill
\end{center}
%
\begin{rem}
\Cref{alg:disc_levy,alg:disc_levy_driv} are in fact a single-level version of a more general coupled discretization scheme given in \cite{DH11} which will be described shortly in \autoref{subsec:MLMC}.
\end{rem}
The Markov transition between $Y_{n-1}$ and $Y_n$, $n=1,\cdots,T_f$, is given by the transition kernel $M^{(\theta,\infty)}(y_{n-1},dy_n)$ of the process in \ref{eq:levy_driv_sde} over unit time. For $n=1,\cdots,T_f$, set $M_n^{(\theta,\infty)}(y_{n-1},dy_n)= M^{(\theta,\infty)}(y_{n-1},dy_n)$, then the pair $(M^{(\theta,\infty)}_n,G_n^{(\theta)} )$ defines a Feynman-Kac type model (see \cite{P04}). In many cases, $M^{(\theta,\infty)}$ exists, but one cannot exactly sample from it or evaluate a non-negative unbiased estimator of it, and therefore, one has to work with a discretization of the model as in \ref{eq:disc_levy_driv}.
\section{Particle Filter and Multilevel Monte Carlo Methods}
\label{sec:coupled_levy}

We mentioned earlier in \autoref{subsec:model} that Step 2. of the main algorithm requires a coupling of the PF. Before discussing that, we give a brief review of the PF algorithm in \autoref{subsec:pf}. In \autoref{subsec:MLMC}, we summarize the MLMC method for diffusions driven by Wiener processes and that of \cite{DH11} for diffusions driven by L\'{e}vy process. Then in \autoref{subsec:cpf_and_ub_level}, we explain the PF coupling of \cite{JKLZ18}, which will be used in our main algorithm to produce an unbiased scheme for estimating the level-difference in \eqref{eq:level_diff}. In all coming subsections we assume that $l\in \mathbb{N}\cup\{0\}$ and $\theta\in\Theta$ are both fixed.

\subsection{Particle Filter}
\label{subsec:pf}
The PF (see for example \cite{P04, DDG01}) is a method that generates a collection of $N$ samples in parallel associated with a set of weights and combines IS and resampling techniques to numerically approximate expectations of functionals w.r.t the normalized or unnormalized measures. Let $\varphi : \mathbb{R}^d \to \mathbb{R}$ be a function in $\mathcal{B}_b(\mathbb{R}^d)$. Suppose that one is interested in approximating the expected value $\eta_n^{(\theta)}(\varphi) =:\mathbb{E}[\varphi(Y_{1:n})|Z_{1:n}=z_{1:n}]$, where $Y_t$ is the solution to the L\'{e}vy-driven SDE in \ref{eq:levy_driv_sde} at discrete time $t\in \mathbb{N}$, and $\eta_n^{(\theta)}(dy_{1:n})$ is the joint probability measure of the observations and the diffusion process defined through \ref{eq:levy_driv_sde} at discrete times, given by
\begin{align}
\label{eq:eta_n_and_gamma_n}
\eta_n^{(\theta)}(dy_{1:n}) \propto \prod_{i = 1}^n G_i^{(\theta)}(y_i) M_i^{(\theta,\infty)}(y_{i-1},dy_i):=\gamma_n^{(\theta)}(dy_{1:n}),
\end{align}
where the dependence on the observations is suppressed. We refer to $\eta_n^{(\theta)}$ and $\gamma_n^{(\theta)}$ as \textit{normalized} and \textit{unnormalized} smoothing measures, respectively. In practice, $M_n^{(\theta,\infty)}(y_{n-1},dy_n)$ is typically intractable as discussed earlier on, and as a result, one considers targets associated to a time discretization given by
\begin{align}
\label{eq:eta_n^l_and_gamma_n^l}
\eta_n^{(\theta,l)}(dy_{1:n}^l) \propto \prod_{i = 1}^n G_i^{(\theta)}(y_i^l) M_i^{(\theta,l)}(y_{i-1}^l,dy_i^l) := \gamma_n^{(\theta,l)}(dy_{1:n}^l),
\end{align}
where $M_n^{(\theta,l)}(y_{n-1}^l,dy_n^l)$ is the induced transition kernel from the Euler approximation in \ref{eq:disc_levy_driv} over unit time, and $\gamma_n^{(\theta,l)}(dy_{1:n}^l)$ refers to the approximate unnormalized measure where $\gamma_n^{(\theta)}= \lim_{l\to \infty
}\gamma_n^{(\theta,l)}$ uniformly in $\theta$. Note that here we include a superscript $l$ on $y_{1:n}$ to emphasize that they are samples at unit times from the discretized SDE in \eqref{eq:disc_levy_driv}. Let 
\begin{align}
\label{eq:eta_n^l(phi)}
\eta_n^{(\theta,l)}(\varphi):= \int_{\mathbb{R}^{nd}} \varphi(y_n^l) ~ \eta_n^{(\theta,l)}(dy_{1:n}^l),
\end{align}
denote the expectation of a bounded and measurable function $\varphi$ w.r.t. the probability measure $\eta_n^{(\theta,l)}(dy_{1:n}^l)$ at discrete time $n$. We also define
\begin{align}
\label{eq:gamma_n^l(phi)}
\gamma_n^{(\theta,l)}(\varphi) :=  \int_{\mathbb{R}^{nd}} \varphi(\tilde{y}_n) ~\gamma_n^{(\theta,l)}(dy^l_{1:n}), 
\end{align}
the expectation of $\varphi$ w.r.t. the unnormalized measure. Note that $\eta_n^{(\theta,l)}(\varphi) = \gamma_n^{(\theta,l)}(\varphi) / \gamma_n^{(\theta,l)}(1)$. Both  $\eta_n^{(\theta,l)}(\varphi)$ and $\gamma_n^{(\theta,l)}(\varphi)$ can be estimated using the PF method described in \autoref{alg:pf}.
%
%
%
%
\begin{center}
\captionsetup[algorithm]{style=algori}
\captionof{algorithm}{Particle filter}
\label{alg:pf}
\raggedright
\begin{enumerate}
\item \textbf{Input:} Level $l$, number of particles $N$, $Y_0$, terminal time $n$, $\theta$ and potential functions $\{G_i^{(\theta)}\}_{i=1}^n$.
\item Initialization: For $i = 1,\cdots,N$ set $\bm{Y_0^{l,i}}=Y_0^{l,i}=Y_0$
\item For $k=1,\cdots,n$, do:
\begin{enumerate}[(i)]
\item  For $i = 1,\cdots,N$ simulate $Y_1^{l,i}$ via \ref{eq:disc_levy_driv} with $Y_0^{l,i}=Y_{k-1}^{l,i}$ and set $\bm{\widehat{Y}_k^{l,i}} = (\bm{Y_{k-1}^{l,i}}, Y_1^{l,i})$.
\item For $i = 1,\cdots,N$ compute $w_k^i = G_k^{(\theta)}(y_1^{l,i})$ and set $W_k^i = w_k^i/\sum_{j=1}^N w_k^j$.
\item Resample the particles $\{\bm{\widehat{Y}_k^{l,i}}\}_{i=1}^N$ according to the normalized weights $\{W_k^i\}_{i=1}^N$. Denote the resampled particles $\{\bm{Y_k^{l,i}}\}_{i=1}^N=\{(Y_0^{l,i},\cdots,Y_k^{l,i})\}_{i=1}^N$. 
\end{enumerate}
\item For $i=1,\cdots,N$ set $V^i = \left[ \prod_{k=1}^{n-1} \frac{1}{N} \sum_{j=1}^N w_k^j \right] \frac{1}{N} w_n^i$.
\item \textbf{Output:} Return $\{\bm{Y_n^{l,i}}, V^i \}_{i=1}^N$.
\end{enumerate}
\vspace{-0.2cm}
\hrulefill
\end{center}
We remark that Step 3(i) of \autoref{alg:pf}  corresponds to the resampling step, which can be multinomial, residual, stratified or systematic (see \cite{DCM05} for a comparison.) In this paper, we use multinomial resampling. It is well-known that $\sum_{i=1}^N V^i ~ \varphi(\bm{Y_n^{l,i}})$ and $\sum_{i=1}^N W_n^i ~\varphi(\bm{Y_n^{l,i}})$ will converge almost surely as $N\to \infty$ to $\gamma_n^{(\theta,l)}(\varphi)$ and $\eta_n^{(\theta,l)}(\varphi)$, respectively. Moreover, the sum $\sum_{i=1}^N V^i ~ \varphi(\bm{Y_n^{l,i}})$ is an unbiased estimator of the unnormalized smoother $\gamma_n^{(\theta,l)}(\varphi)$ (see e.g. \cite{P04, RJ20}). 

A quantity of interest that one aims to estimate is the level difference $\eta_n^{(\theta,l)}(\varphi) - \eta_n^{(\theta,l-1)}(\varphi) $ (or $\gamma_n^{(\theta,l)}(\varphi) - \gamma_n^{(\theta,l-1)}(\varphi) $) for some $l\in\mathbb{N}$. This quantity appears in multiple algorithms like those found in the ML context (e.g. \cite{DH11, Giles08, MLPF17}) and in some debiasing schemes (e.g. \cite{M11, RG15}). In particular, it appears in the debiasing scheme within the correction step of the main algorithm in \autoref{sec:ub_infer}.
In \autoref{subsec:MLMC}, we explain how to obtain coupled (possibly biased) samples from the law of the discretized L\'{e}vy-driven process in \ref{eq:disc_levy_driv}. Then in \autoref{subsec:cpf_and_ub_level}, we present an unbiased estimator of the level difference
\begin{align}
\label{eq:level_diff}
\gamma_n^{(\theta,l)}(\varphi) - \gamma_n^{(\theta,l-1)}(\varphi),
\end{align}
for given $(\theta, l) \in \Theta \times \mathbb{N}$ and function $\varphi \in \mathcal{B}_b(\mathbb{R}^d)$. Before moving forward, we require the following assumptions on the Feynman-Kac model defined by $(M_n^{(\theta,l)}, G_n^{(\theta)})$. 
\begin{ass}
\label{ass:Gn_and_Mn}
There are $c > 1$ and $C > 0$, such that for all $n,l\in \mathbb{N}$, $\varphi \in\textrm{Lip}_{\|\cdot\|_2}(\mathbb{R}^d) \cap \mathcal{B}_b(\mathbb{R}^d)$, and $(\theta, x,x')\in \Theta \times \mathbb{R}^d \times \mathbb{R}^d$, we have
\begin{enumerate}[(i)]
\item $c^{-1} \leq G_n^{(\theta)}(x) \leq c$.
\item $G_n^{(\theta)}$ is Lipschitz, i.e., $\|G_n^{(\theta)}(x)-G_n^{(\theta)}(x') \|_2 \leq C \|x-x'\|_2$.
\item $M_n^{(\theta,l)}(\varphi)$ is Lipschitz, i.e., $\|M_n^{(\theta,l)}(\varphi)(x)-M_n^{(\theta,l)}(\varphi)(x') \|_2 \leq C \|x-x'\|_2$.
\end{enumerate}
\end{ass}
\subsection{A Multilevel Monte Carlo Method}
\label{subsec:MLMC}
First, let us assume that $\{Y_t\}_{t\in [0,1]}$ is a diffusion process that is driven by a Wiener process. Suppose that we are interested in estimating the expectation $\mathbb{E}[\varphi(Y_1)]:=\zeta_1(\varphi)$, where $\zeta_1$ is the law of the process $\{Y_t\}_{t\in [0,1]}$ at time 1. Furthermore, assume that one has to work with a discretization of the the diffusion process. Let $L\in \mathbb{N}$ be the level of discretization and define $\zeta_1^L(\varphi) := \mathbb{E}[\varphi(Y_1^L)]$ the expectation w.r.t. the law $\zeta_1^L$ associated with the discretization scheme at time 1. The standard Monte Carlo (MC) provides an unbiased estimator of this expectation through 
$$
\zeta_1^{L,MC}(\varphi) := \frac{1}{N} \sum_{i=1}^N \varphi(Y_1^{L,i}),
$$
where $\{Y_1^{L,i}\}_{i=1}^N$ are i.i.d. samples from $\zeta_1^L$. Then the mean square error (MSE) of the estimator is
\begin{align*}
\mathbb{E}\left[\left(\zeta_1^{L,MC}(\varphi) - \zeta_1(\varphi) \right)^2 \right] &=\mathbb{E}\left[\left(\zeta_1^{L,MC}(\varphi) -  \zeta_1^L(\varphi) + \zeta_1^L(\varphi)- \zeta_1(\varphi) \right)^2 \right] \\
&= \frac{\mathbb{V}ar[\varphi(Y_1^L)]}{N} +  \left( \zeta_1^L(\varphi)- \zeta_1(\varphi) \right)^2,
\end{align*}
where the first term on the right is the variance and the second term is the bias squared which is of order $\mathcal{O}(\Delta_L^2)$. To get an MSE of order $\mathcal{O}(\epsilon^2)$, for a given $\epsilon>0$, one needs to choose $N=\mathcal{O}(\epsilon^{-2})$ and $\Delta_L=\mathcal{O}(\epsilon)$. If we assume that the cost of sampling $Y_1^L$ from $\zeta_1^L$ is $C_L = \mathcal{O}(\Delta_L^{-1})$, then to achieve an MSE of order $\mathcal{O}(\epsilon^2)$, the cost of the MC estimator will be $\mathcal{O}(N\Delta_L^{-1})=\mathcal{O}(\epsilon^{-3})$. On the other hand, MLMC, which is proposed in \cite{Giles08} (see also \cite{MLPF17}) for diffusions driven by a Wiener process, achieves this MSE at a lower cost. The idea of MLMC method is to write $\zeta_1^L(\varphi)$ as a telescoping sum
\begin{align*}
\zeta_1^L(\varphi) = \zeta_1^0(\varphi) + \sum_{l=1}^L [\zeta_1^l(\varphi) -  \zeta_1^{l-1}(\varphi)].
\end{align*}
The sum is estimated by approximating $\zeta_1^0(\varphi)$ through a standard MC method with $N_0$ independent samples, and approximating the differences $\zeta_1^l(\varphi) - \zeta_1^{l-1}(\varphi)$ by independently simulating $N_l$ samples $\{Y_1^{l,i}, Y_1^{l-1,i}\}_{i=1}^{N_l}$ from an appropriate coupling of the diffusion process. Let $\{Y_t\}_{t\in [0,1]}$ be the solution of an SDE driven by a Wiener process and approximated with an Euler discretization, then the coupling emerges from concatenating the fine increments of the Wiener process to generate the coarse increments. The MLMC estimator of $\zeta_1^L(\varphi)$ is then
$$
\zeta_1^{L,MLMC} := \frac{1}{N_0} \sum_{i=1}^{N_0} \varphi(Y_1^{0,i}) +\sum_{l=1}^L  \frac{1}{N_l} \sum_{i=1}^{N_l} [\varphi(Y_1^{l,i}) - \varphi(Y_1^{l-1,i})].
$$ 
Similar to before, the overall MSE for the ML estimator $\zeta_1^{L,MLMC}$ can be expressed as
$$
\mathbb{E}\left[\left(\zeta_1^{L,MLMC}(\varphi) - \zeta_1(\varphi) \right)^2 \right] =  \sum_{l=0}^L \frac{\mathbb{V}ar[\varphi(Y_1^l)-\varphi(Y_1^{l-1})]}{N_l} + \left( \zeta_1^L(\varphi)- \zeta_1(\varphi) \right)^2,
$$
with the convention $\varphi(Y_1^{-1})=0$. The sum on the right hand side is the variances and the other term is the bias squared. Let $V_l=\mathbb{V}ar[\varphi(Y_1^l)-\varphi(Y_1^{l-1})]$, then the dependent coupling aforementioned will induce the variance to be a function of $\Delta_l$. MLMC aims at reducing the variance term $\sum_{l=0}^L V_l/N_l$, leaving unchanged the bias due to the Euler discretization. The level $L$ and the number of samples $\{N_l\}_{l=0}^L$ are chosen to balance these terms so that the MSE is of $\mathcal{O}(\epsilon^2)$, for a given $\epsilon>0$, and that the overall computational complexity is minimized. In particular, $L$ is chosen so that the bias term is $\mathcal{O}(\epsilon^2)$, and given $C_l$ and $V_l$ as functions of $\Delta_l$, the total cost $\sum_{l=0}^L C_l N_l$ is minimized with the constrain $\sum_{l=0}^L V_l/N_l=\mathcal{O}(\epsilon^2)$ by optimizing $\{N_l\}_{l=0}^L$. In \cite{Giles08} the author shows that the cost to achieve such an MSE using an MLMC estimator is
\begin{align*}
\text{Cost} = \left\{ 
\begin{array}{lll}
\mathcal{O}(\epsilon^{-2}\log(\epsilon)^2) & \text{if} & \beta = 1,\\
\mathcal{O}(\epsilon^{-2}) & \text{if} & \beta > 1,\\
\mathcal{O}(\epsilon^{-2-(1-\beta)/\alpha}) & \text{if} & 0<\beta < 1,\\
\end{array} \right.
\end{align*}
where $\alpha$ and $\beta$ are respectively the weak and strong error rates of the discretization scheme.

Nevertheless, the situation is more complicated when SDEs are purely driven by general L\'{e}vy processes. An approach based on Poisson thinning has been suggested in \cite{GX12} for pure-jump diffusion and by \cite{FKSS14} for general L\'{e}vy processes. However, an alternative construction based on the L\'{e}vy-It\^{o} decomposition is proposed in \cite{DH11}. The authors of that article provide a procedure to obtain coupled increments for the L\'{e}vy process on $[0,1]$ which then can be used in \ref{eq:disc_levy_driv} to generate samples $\{Y_1^{l,i}, Y_1^{l-1,i}\}_{i=1}^{N_l}$ from a coupled-kernel induced by the Euler scheme. This construction is employed in our paper and we summarize it in \Cref{alg:coupled_disc_levy,alg:disc_coupled_levy_driv}. We introduce the following necessary notations.
\begin{definition}
\label{def:coupling}
Let $l\in \mathbb{N}$, $\theta\in \Theta$ and $\check{y}_n^l = (y_n^l,y_n^{l-1})$. 
Define the kernel $M_n^{(\theta,l,l-1)}:[\mathbb{R}^d \times \mathbb{R}^d] \times [\sigma(\mathbb{R}^d) \otimes \sigma(\mathbb{R}^d)]\to \mathbb{R}_+$
a coupling of the transition kernels $M_n^{(\theta,l)}(y_n^l,\cdot)$ and $M_n^{(\theta,l-1)}(y_n^{l-1},\cdot)$ such that for any $A\in \sigma(\mathbb{R}^d)$, we have
\begin{align*}
&\int_{A \times \mathbb{R}^d} M_n^{(\theta,l,l-1)}(\check{y}_{n-1}^l,d\check{y}_n^l) = M_n^{(\theta,l)}(y_{n-1}^l,A) \quad \text{and}\quad \\
& \int_{ \mathbb{R}^d \times A} M_n^{(\theta,l,l-1)}(\check{y}_{n-1}^l,d\check{y}_n^l) = M_n^{(\theta,l-1)}(y_{n-1}^{l-1},A).
\end{align*}
In addition, for $(y,y') \in\mathbb{R}^d \times \mathbb{R}^d$ and $\varphi\in \mathcal{B}_b(\mathbb{R}^d \times \mathbb{R}^d)$, we introduce the following notation
$$
M_n^{(\theta,l,l-1)}(\varphi)(y,y') :=  \int_{\mathbb{R}^{2d}} \varphi(\check{y}_n^l) M_n^{(\theta,l,l-1)}((y,y'),d\check{y}_n^l).
$$
\end{definition}
%
%
\begin{center}
\captionsetup[algorithm]{style=algori}
\captionof{algorithm}{Coupled increments of the L\'{e}vy process on $[0,1]$}
\label{alg:coupled_disc_levy}
\raggedright
\begin{enumerate}
\item \textbf{Input:} $l$, $\delta_{l-1}$, $\lambda_l$, $\mu^l(dx)$ and $\theta$.
\item Run \autoref{alg:disc_levy} with input $l$, $\lambda_l$, $\mu^l(dx)$ and $\theta$ to return $K^l$, $\{T_i^l\}_{i=0}^{K^l}$ and $\{\Delta L_{T_i^l}^l\}_{i=1}^{K^l}$. Set $\widehat{K}^{l-1}=0$.
\item Generate jump times and heights at level $l-1$: For $ i = 1, \cdots, K^l$, if $|\Delta L_{T_i^l}^l| \geq \delta_{l-1}$, set $\widehat{T}_i^{l-1} = T_i^l$, $\Delta L_{\widehat{T}_i^{l-1}}^{l-1} = \Delta L_{T_i^l}^l$ and $ \widehat{K}^{l-1} = \widehat{K}^{l-1} + 1$.
\item Refine jump times: Set $i=1$, and $T_0^{l-1}=0$.
\begin{itemize}
\item[(i)] If $\Big\{\widehat{T}_j^{l-1}>T_{i-1}^{l-1}~:~ j\in\{1,\cdots,\widehat{K}^{l-1}\} \Big\}\neq \phi$, then  $T_i^{l-1} = \min \Bigg\{T_{i-1}^l+\Delta_{l-1},~ \min \Big\{\widehat{T}_j^{l-1}>T_{i-1}^{l-1}~:~ j\in\{1,\cdots,\widehat{K}^{l-1}\} \Big\} \Bigg\}$; otherwise $T_i^{l-1} = \min \Big\{T_{i-1}^l+\Delta_{l-1},~ 1 \Big\}$.
\item[(ii)] If $T_i^{l-1} = \widehat{T}_j^{l-1}$ for some $j\in \{1,\cdots,\widehat{K}^{l-1}\}$, set $\Delta L_{T_i^{l-1}} =  \Delta L_{\widehat{T}_j^{l-1}}$; otherwise $\Delta L_{T_i^{l-1}} =  0$.
\item[(iii)] if $T_i^{l-1} =1$, set $K^{l-1} = i$; otherwise $i=i+1$ and go to (i).
\end{itemize}
\item \textbf{Output:} Return $K^l$, $K^{l-1}$, $\{T_i^l\}_{i=1}^{K^l}$, $\{T_i^{l-1}\}_{i=1}^{K^{l-1}}$, $\{\Delta L_{T_i^l}^l\}_{i=1}^{K^l}$ and $\{\Delta L_{T_i^{l-1}}^{l-1}\}_{i=1}^{K^{l-1}}$.
\end{enumerate}
\vspace{-0.1cm}
\hrulefill
\end{center}
%
\autoref{alg:disc_coupled_levy_driv} produces samples from the coupled-kernel $M_n^{(\theta,l,l-1)}$ associated with the discretization in \eqref{eq:disc_levy_driv}. By \cite[Theorem 2]{DH11}, there exists a $0 < C < \infty$ such that for any $(y,\theta,l) \in \mathbb{R}^d\times\Theta\times \mathbb{N}$
\begin{align}
\label{eq:DH_Thm}
\int_{\mathbb{R}^{2d}} |y_n^l - y_n^{l-1} |^2 M_n^{(\theta,l,l-1)}((y,y),d\check{y}_n^l) \leq C \Delta_l^\beta,
\end{align}
where $\beta$ is the strong error rate of the Euler scheme presented in \autoref{subsec:Euler_sde} that can be obtained from the bounds in \eqref{eq:strong_err_Euler}. Moreover, in \cite[Theorem 1]{DH11} the author show that if $\int (\frac{|x|^2}{\delta_l} \wedge 1) ~\nu(dx) \leq g(\delta_l)$ for some decreasing and invertible function $g:(0,\infty) \to (0,\infty)$ and all $\delta_l>0$, then the MSE $:=\mathbb{E}[(\zeta_1^{L,MLMC}(\varphi) - \zeta_1(\varphi))^2 ]$ satisfies the following bounds:
\begin{enumerate}[(i)]
\item If $g(\delta_l) \leq 1/(\delta_l \log(1/\delta_l)^{1+\gamma})$ as $\delta_l\to 0 $ for some $\gamma >0$, and $\Sigma = 0$, then the Cost $\leq C$, for some $C>0$ and
\begin{align*}
\text{MSE} \leq 1/C.
\end{align*}
\item If $g(\delta_l) \leq \log(1/\delta_l)^{\gamma}/\delta_l $ as $\delta_l\to 0 $ for some $\gamma\geq 1/2$, then the Cost $\leq C$, for some $C>0$ and
\begin{align*}
\text{MSE} \leq \frac{1}{C} \log(C)^{2(1+\gamma)}.
\end{align*}
\item If there exists some $\gamma >0$ such that $g(\delta_l) \leq \frac{1}{2} g(\gamma\delta_l/2)$ for all sufficiently small $\delta_l>0$, then the Cost $\leq C$, for some $C>0$ and
\begin{align*}
\text{MSE} \leq C [g^{-1}(C)]^2.
\end{align*}
\end{enumerate}

%
%
\begin{center}
\captionsetup[algorithm]{style=algori}
\captionof{algorithm}{Coupled sampling for L\'{e}vy-driven SDEs}
\label{alg:disc_coupled_levy_driv}
\raggedright
\begin{enumerate}
\item \textbf{Input:} $Y_{0}^l$, $Y_0^{l-1}$, $l$, $\lambda_l$, $F_0^l$, $\delta_{l-1}$, $F_0^{l-1}$, $\mu^l(dx)$ and $\theta$.
\item Call \autoref{alg:coupled_disc_levy} to return $K^l$, $K^{l-1}$, $\{T_i^l\}_{i=1}^{K^l}$, $\{T_i^{l-1}\}_{i=1}^{K^{l-1}}$, $\{\Delta L_{T_i^l}^l\}_{i=1}^{K^l}$ and $\{\Delta L_{T_i^{l-1}}^{l-1}\}_{i=1}^{K^{l-1}}$. 
\item For $i=1, \cdots,K^l$, sample $W_{T_i^l}$ from $\mathcal{N}_r(0,T_i^l I_r)$. Notice that $\{W_{T_i^{l-1}}\}_{i=1}^{K^{l-1}} \subset \{W_{T_i^{l}}\}_{i=1}^{K^{l}}$.
\item With $\{T_i^l\}_{i=1}^{K^l}$, $\{\Delta L_{T_i^l}^l\}_{i=1}^{K^l}$, $\{W_{T_i^l}-W_{T_{i-1}^l}\}_{i=1}^{K^l}$ and $F_0^l$, generate the increments $\{\Delta X_{T_i^l}^l\}_{i=1}^{K^l}$ through \ref{eq:disc_levy}. With $\{T_i^{l-1}\}_{i=1}^{K^{l-1}}$, $\{\Delta L_{T_i^{l-1}}^{l-1}\}_{i=1}^{K^{l-1}}$, $\{W_{T_i^{l-1}}-W_{T_{i-1}^{l-1}}\}_{i=1}^{K^{l-1}}$ and $F_0^{l-1}$, generate the increments $\{\Delta X_{T_i^{l-1}}^{l-1}\}_{i=1}^{K^{l-1}}$ via \ref{eq:disc_levy}.
\item Given the coupled increments $\{\Delta X_{T_i^l}^l\}_{i=1}^{K^l}$ and $\{\Delta X_{T_i^{l-1}}^{l-1}\}_{i=1}^{K^{l-1}}$, compute  $Y_1^l$ and $Y_1^{l-1}$ with $Y_{0}^l$ and $Y_0^{l-1}$, respectively, via the recursion in \ref{eq:disc_levy_driv}.
\item \textbf{Output:} Return $Y_1^l$ and $Y_1^{l-1}$.
\end{enumerate}
\vspace{-0.2cm}
\hrulefill
\end{center}
%
%

\subsection{Coupled PF and Unbiased Level Difference Estimation}
\label{subsec:cpf_and_ub_level}
As assumed previously, let $\varphi \in \mathcal{B}_b(\mathbb{R}^d)$ and $\theta\in \Theta$ be fixed. Consider the level difference in \eqref{eq:level_diff}. If one can estimate the difference using a coupling of $(\gamma_n^l, \gamma_n^{l-1})$, it is then possible to obtain a variance reduction. As seen in the previous subsection this is a crucial requirement in the ML estimation, where one wants the variance of the differences to decay as $l$ increases. The authors of \cite{JKLZ18} introduced an effective general coupling of PFs which we use to unbiasedly estimate the level difference in \ref{eq:level_diff}. Coupled particle filter (CPF) is usually used in multilevel PF (MLPF) algorithms (see \cite{MLPF17} and \cite{AKP19} for more details on MLPF for models driven by Wiener and  L\'{e}vy processes, respectively) or debiasing schemes (see e.g. \cite{CFJLV21,RJ20}), similar to MLMC, one writes the expectation of some function w.r.t. smoothing distribution as a telescoping sum, then independently run a particle filter at the first level and CPFs for the higher levels. In \autoref{alg:CPF} we present the CPF used in \cite{MLPF17,AKP19} which will generate samples from the coupling of $(\gamma_n^l, \gamma_n^{l-1})$.
\begin{definition}
\label{def:G_check}
Let $l\in \mathbb{N}$, $\theta\in \Theta$ and $\check{y}_n^l = (y_n^l,y_n^{l-1})$. Define the potential 
$$
\widecheck{G}_n^{(\theta)}(\check{y}^l) := \frac{1}{2} (G_n^{(\theta)}(y_n^l) + G_n^{(\theta)}(y_n^{l-1}) ).
$$
\end{definition}
%
%
%
%
%
\begin{center}
\captionsetup[algorithm]{style=algori}
\captionof{algorithm}{Coupled particle filter}
\label{alg:CPF}
\raggedright
\begin{enumerate}
\item \textbf{Input:} Level $l$, number of particles $N_l$, $Y_0$, time $n$, $\theta$,  $\lambda_l$, $F_0^l$, $\delta_{l-1}$, $F_0^{l-1}$, $\mu^l(dx)$ and potential functions $\{\widecheck{G}_k^{(\theta)}\}_{k=1}^n$.
\item Initialization: For $i = 1,\cdots,N$ set $\bm{Y_0^{l,i}}=\bm{Y_0^{l-1,i}}=Y_0^{l,i}=Y_0^{l-1,i}=Y_0$.
\item For $k=1,\cdots,n$, do:
\begin{enumerate}[(i)]
\item  For $i = 1,\cdots,N_l$ call \autoref{alg:disc_coupled_levy_driv} with $Y_0^{l,i}=Y_{k-1}^{l,i}$ and $Y_0^{l-1,i}=Y_{k-1}^{l-1,i}$ in the input to return $Y_1^{l,i}$ and $Y_1^{l-1,i}$ and set $\bm{\widehat{Y}_k^{l,i}} = (\bm{Y_{k-1}^{l,i}}, Y_1^{l,i})$ and $\bm{\widehat{Y}_k^{l-1,i}} = (\bm{Y_{k-1}^{l-1,i}}, Y_1^{l-1,i})$.
\item For $i = 1,\cdots,N_l$ compute $w_k^i = \widecheck{G}_k^{(\theta)}(\check{y}_k^{l,i})$ and set $W_k^i = w_k^i/\sum_{j=1}^{N_l} w_k^j$.
\item Resample the particles $\{\bm{\widehat{Y}_k^{l,i}}\}_{i=1}^N$ and $\bm{\widehat{Y}_k^{l-1,i}}$ according to the normalized weights $\{W_k^i\}_{i=1}^{N_l}$ by sampling the particles indices from a multinomial distribution. Denote the resampled particles $\{\bm{Y_k^{l,i}}\}_{i=1}^{N_l}=\{(Y_0^{l,i},\cdots,Y_k^{l,i})\}_{i=1}^{N_l}$ and $\{\bm{Y_k^{l-1,i}}\}_{i=1}^N=\{(Y_0^{l-1,i},\cdots,Y_k^{l-1,i})\}_{i=1}^N$.
\end{enumerate}
\item For $i=1,\cdots,N_l$ set $\widecheck{V}^i = \left[ \prod_{k=1}^{n-1} \frac{1}{N_l} \sum_{j=1}^{N_l} w_k^j \right] \frac{1}{N_l} w_n^i$.
\item \textbf{Output:} Return $\{\bm{\widecheck{Y}_n^i}\}_{i=1}^{N_l}=\{ (\bm{Y_n^{l,i}},\, \bm{Y_n^{l-1,i}})\}_{i=1}^{N_l}$ and $\{\widecheck{V}^i \}_{i=1}^{N_l}$.
\end{enumerate}
\vspace{-0.2cm}
\hrulefill
\end{center}
We then use the output of \autoref{alg:CPF} in \autoref{alg:ub_level_diff} to unbiasedly estimate the level differences in \ref{eq:level_diff}. 
%
%
\begin{center}
\captionsetup[algorithm]{style=algori}
\captionof{algorithm}{Unbiased estimation of level differences}
\label{alg:ub_level_diff}
\raggedright
\begin{enumerate}
\item \textbf{Input:} Level $l$, number of particles $N_l$, $Y_0$, time $n$, $\theta$,  $\lambda_l$, $F_0^l$, $\delta_{l-1}$, $F_0^{l-1}$, $\mu^l(dx)$ and potential functions $\{G_k^{(\theta)},\widecheck{G}_k^{(\theta)}\}_{k=1}^n$.
\item Call \autoref{alg:CPF} to return $\{\bm{\widecheck{Y}_n^i}\}_{i=1}^{N_l}=\{ (\bm{Y_n^{l,i}},\, \bm{Y_n^{l-1,i}})\}_{i=1}^{N_l}$ and $\{\widecheck{V}^i \}_{i=1}^{N_l}$. Then perform the following
\begin{enumerate}[(i)]
\item For $i=1,\cdots,N_l$, set $\bm{Y_n^i} = \bm{Y_n^{l,i}}$ and $V^i = \widecheck{V}^i H^{(\theta,l)}(\check{y}_n^{l,i})$, where 
\begin{align}
H^{(\theta,l)}(\check{y}_n^{l,i}) = \prod_{k=1}^n \frac{G_k^{(\theta)}(y_k^{l,i})}{\widecheck{G}_k^{(\theta)}(\check{y}_k^{l,i})}.
\end{align}
\item For $i=N_l+1,\cdots,2N_l$, set $\bm{Y_n^i} = \bm{Y_n^{l-1,i-N_l}}$ and $V^i = -\widecheck{V}^{i-N_l} H^{(\theta,l-1)}(\check{y}_n^{l,i-N_l})$, where 
\begin{align}
H^{(\theta,l-1)}(\check{y}_n^{l,i-N_l}) = \prod_{k=1}^n \frac{G_k^{(\theta)}(y_k^{l-1,i-N_l})}{\widecheck{G}_k^{(\theta)}(\check{y}_k^{l,i-N_l})}.
\end{align}
\end{enumerate}
\item \textbf{Output:} Return $\{\bm{Y}_n^{l,i} ,\, V^{l,i}\}_{i=1}^{2N_l}$.
\end{enumerate}
\vspace{-0.2cm}
\hrulefill
\end{center}
%
Since the PF method provides an unbiased estimator of the unnormalized smoothing density, one can easily use that to show (see Proposition 2.4 in \cite{CFJLV21})
$$
\mathbb{E}\left[\sum_{i=1}^{2N_l} V^i \, \varphi(\bm{Y}_n^{l,i}) \right]= \gamma_n^{(\theta,l)}(\varphi) - \gamma_n^{(\theta,l-1)}(\varphi).
$$ 
\begin{rem}
\begin{enumerate}[(i)]
\item The choice of $\widecheck{G}_n^{(\theta)} $ in this paper is similar to that in \cite{CFJLV21}. However, in the original work \cite{JKLZ18} where this methodology was employed, the authors take
$$
\widecheck{G}_n^{(\theta)}(\check{y}_n^l) =\max\{ G_n^{(\theta)}(y_n^l), G_n^{(\theta)}(y_n^{l-1}) \}.
$$
In fact, it is sufficient for unbiasedness to choose $\widecheck{G}_n^{(\theta)}$ such that $H^{(\theta,l)}$ and $H^{(\theta,l-1)}$ are uniformly upper bounded; this way the variance of the weights is independent of the length of the observations $n$, hence, does not grow as $n\to \infty$.
\item The unbiasedness of the level-difference is achieved regardless of the number of particles.

\item We remark that one can combine \Cref{alg:pf,alg:ub_level_diff} along with the randomization techniques in \cite{M11,RG15} to provide an unbiased estimator of $\gamma_n^{(\theta)}(\varphi)$, with $\gamma_n^{(\theta)}$ as in \eqref{eq:eta_n_and_gamma_n}. This leads to unbiased inference w.r.t. the normalized smoother $\eta_n^{(\theta)}(\varphi)$ (see \cite[Proposition 2.8]{CFJLV21}.)
\end{enumerate}

\end{rem}
We state the following Lemma which some of our claims we will rely on.
\begin{lem}
\label{lem:bound_on_M}
Under \Cref{ass:f_and_v,ass:delta_l}, there exists a $C< \infty$ such that for any $l\in \mathbb{N}$ and $\theta \in \Theta$,
\begin{align*}
\sup_{\varphi\in \mathcal{A}} \sup_{y\in\mathbb{R}^d} \left|M_n^{(\theta,l)}(\varphi)(y)- M_n^{(\theta,l-1)}(\varphi)(y) \right| &\leq C \Delta_l^{\beta/2} 
\end{align*}
where $\mathcal{A}=\mathcal{B}_b(\mathbb{R}^{d}) \cap \textrm{Lip}_{\|\cdot\|_2}( \mathbb{R}^{d})$ and $\beta$ is the strong error rate of the discretization scheme for the SDE in \eqref{eq:levy_driv_sde}. 
\end{lem}
\begin{proof}
\begin{align*}
&\Big|M_n^{(\theta,l)}(\varphi)(y)-  M_n^{(\theta,l-1)}(\varphi)(y) \Big|\\ &= \left| M_n^{(\theta,l,l-1)}(\varphi \otimes 1)(y,y) - M_n^{(\theta,l,l-1)}(1\otimes \varphi)(y,y) \right| \quad \text{by \autoref{def:coupling}}\\
&=\Bigg|\int_{\mathbb{R}^{2d}} \varphi(y_n^l) M_n^{(\theta,l,l-1)}((y,y),d\check{y}_n^l) - \int_{\mathbb{R}^{2d}} \varphi(y_n^{l-1}) M_n^{(\theta,l,l-1)}((y,y),d\check{y}_n^l)  \Bigg|\\
&=\Bigg| \int_{\mathbb{R}^{2d}} (\varphi(y_n^l) - \varphi(y_n^{l-1}) ) M_n^{(\theta,l,l-1)}((y,y),d\check{y}_n^l)  \Bigg| \\
&\leq \Bigg( \int_{\mathbb{R}^{2d}} (\varphi(y_n^l) - \varphi(y_n^{l-1}) )^2 M_n^{(\theta,l,l-1)}((y,y),d\check{y}_n^l)  \Bigg)^{1/2} \quad \text{by Jensen's inequality} \\
& \leq \Bigg( \int_{\mathbb{R}^{2d}} |y_n^l - y_n^{l-1} |^2 M_n^{(\theta,l,l-1)}((y,y),d\check{y}_n^l)  \Bigg)^{1/2}  \qquad \text{since }\varphi \in \textrm{Lip}_{\|\cdot\|_2}( \mathbb{R}^{d}).
\end{align*}
Applying the result in \eqref{eq:DH_Thm} in the above inequality concludes the proof.
\end{proof}
\begin{prop}
\label{prop:bound_on_level_diff}
Under \autoref{ass:Gn_and_Mn} (i)-(ii), for any $(\theta,l)\in \Theta\times\mathbb{N}$, $n\geq 0$ and $\varphi \in \textrm{Lip}_{\|\cdot\|_2}(\mathbb{R}^{nd}) \cap \mathcal{B}_b(\mathbb{R}^{nd})$ there exists a $C <\infty$ such that
\begin{align}
\label{eq:bound_on_level_diff}
\left|\gamma_n^{(\theta,l)}(\varphi) - \gamma_n^{(\theta,l-1)}(\varphi)\right| \leq C \Delta_l^{\beta/2}
\end{align}
\end{prop}
\begin{proof}
Since \autoref{ass:Gn_and_Mn} (i)-(ii) holds for all $\theta\in \Theta$, the bound in \eqref{eq:bound_on_level_diff} is uniform in $\theta$, therefore one can adopt the proof in \cite{CFJLV21}.
\end{proof}
\begin{theorem}
\label{thm:main_thm}
Under \autoref{ass:Gn_and_Mn}, for any $(\theta,l)\in \Theta\times\mathbb{N}$ and $\varphi\in \textrm{Lip}_{\|\cdot\|_2}(\mathbb{R}^{nd}) \cap \mathcal{B}_b(\mathbb{R}^{nd})$,  there exists a $C<\infty$ such that 
\begin{align}
\mathbb{E}\left[\left( \sum_{i=1}^{2N_l} V^i \, \varphi( \bm{Y}_n^{l,i}) -\left\{ \gamma_n^{(\theta,l)}(\varphi) - \gamma_n^{(\theta,l-1)}(\varphi)\right\} \right)^2\right]\leq C \frac{\Delta_l^{2 \wedge \beta}}{N_l},
\end{align}
where $\{\bm{Y}_n^{l,i} ,\, V^{l,i}\}_{i=1}^{2N_l}$ is the output of \autoref{alg:ub_level_diff} for the given $\theta$.
\end{theorem}
\begin{proof}
The proof is essentially the same as in \cite{CFJLV21}, except that one needs to use \autoref{lem:bound_on_M} and inequality \eqref{eq:DH_Thm} along the lines. 
\end{proof}
Note that 
\begin{align*}
&\mathbb{E}\left[\left( \sum_{i=1}^{2N_l} V^i \, \varphi(\bm{Y}_n^{l,i})\right)^2\right] \\
&= \mathbb{E}\left[\left( \sum_{i=1}^{2N_l} V^i \, \varphi(\bm{Y}_n^{l,i}) - \left\{ \gamma_n^{(\theta,l)}(\varphi) - \gamma_n^{(\theta,l-1)}(\varphi)\right\} + \left\{ \gamma_n^{(\theta,l)}(\varphi) - \gamma_n^{(\theta,l-1)}(\varphi)\right\} \right)^2\right]\\
&\leq \mathbb{E}\left[\left( \sum_{i=1}^{2N_l} V^i \, \varphi(\bm{Y}_n^{l,i}) -\left\{ \gamma_n^{(\theta,l)}(\varphi) - \gamma_n^{(\theta,l-1)}(\varphi)\right\} \right)^2\right]\\
&\qquad+ \mathbb{E}\left[\left|  \gamma_n^{(\theta,l)}(\varphi) - \gamma_n^{(\theta,l-1)}(\varphi) \right|^2\right].
\end{align*}
By \autoref{thm:main_thm} and \autoref{prop:bound_on_level_diff}, we then have
\begin{align}
\label{eq:bound_on_pf_level_diff}
\mathbb{E}\left[\left( \sum_{i=1}^{2N_l} V^i \, \varphi(\bm{Y}_n^{l,i})\right)^2\right] \leq C \left( \frac{\Delta_l^{2 \wedge \beta}}{N_l} +  \Delta_l^{\beta} \right).
\end{align}
\begin{rem}
\autoref{thm:main_thm} and the discussion right after can be easily established for functions $\varphi: \textrm{Lip}_{\|\cdot\|_2}(\Theta \times\mathbb{R}^{nd}) \cap \mathcal{B}_b(\Theta \times\mathbb{R}^{nd})\to \mathbb{R}$.
\end{rem}

\section{Overall Methodology}
\label{sec:ub_infer}
In this section, we provide the main algorithm of this paper, \autoref{alg:ub_infer}. As previously discussed, the goal is to perform an unbiased inference for the Bayesian model posterior associated to the HMM model in \ref{eq:levy_driv_sde}-\ref{eq:obs}, which is given by
\begin{align}
\label{eq:posterior}
\Pi(d\theta,dy_{1:n}) \propto \overline{\Pi}(d\theta) \prod_{p=1}^n G_p^{(\theta)}(y_p) M_p^{(\theta,\infty)}(y_{p-1},dy_p)~=\overline{\Pi}(d\theta)\gamma_n^{(\theta)}(dy_{1:n}),
\end{align}
where $\theta$ is the static parameter to infer with a prior measure $\overline{\Pi}(d\theta)$, $y_{1:n}$ are the hidden states and $\gamma_n^{(\theta)}(dy_{1:n})$ is the unnormalized filter density defined in \eqref{eq:eta_n_and_gamma_n}. We assume the transition densities $M_n^{(\theta,\infty)}$ cannot be simulated directly, but there are discretized densities $\{M_n^{(\theta,l)}\}_{l\in \mathbb{N}}$ that approximate $M_n^{(\theta,\infty)}$ which one can simulate from. Therefore, we work with the following posterior instead
\begin{align}
\label{eq:disc_posterior}
\Pi^{(l)}(d\theta,dy_{1:n}) \propto \overline{\Pi}(d\theta) \prod_{p=1}^n G_p^{(\theta)}(y_p) M_p^{(\theta,l)}(y_{p-1},dy_p)~=\overline{\Pi}(d\theta)\gamma_n^{(\theta,l)}(dy_{1:n}).
\end{align}
Given some function $\varphi$ that is bounded and measurable on $\Theta\times\mathbb{R}^{nd}$, we are interested in unbiasedly estimating the expectation $\Pi(\varphi)$. \autoref{alg:ub_infer} provides the methodology to do that. The algorithm consists of two main parts; the first is a direct implementation of PMMH \cite{ADH10} at the coarsest level $l_{\text{min}}$, e.g. $l_{\text{min}}=1$, whilst the second part is just running the unbiased level difference scheme in \autoref{alg:ub_level_diff}, conditionally upon the results from the first part at random levels $l$ sampled from a given PMF $P_l = \{p_l\}_{l\in \{l_{\text{min}}+1,\cdots\}}$. The first part of the algorithm does not yield unbiased inference, however the second part does \cite{CFJLV21}. 

%
%
%
\begin{center}
\captionsetup[algorithm]{style=algori}
\captionof{algorithm}{Unbiased estimator of $\Pi(\varphi)$}
\label{alg:ub_infer}
\raggedright
\begin{enumerate}
\item \textbf{Input:} A function $\varphi$ defined on $\Theta \times \mathbb{R}^{nd}$, the number of particles $\{N_l\}_{l\geq l_{\text{min}}}$, initial value $Y_0$, time $n$, $\mu^l(dx)$, $\lambda_l$, $F_0^l$, $\delta_{l-1}$, $F_0^{l-1}$, $\epsilon>0$, potential functions $\{G_k^{(\theta)},\widecheck{G}_k^{(\theta)}\}_{k=1}^n$, a prior $\overline{\Pi}(d\theta)$, a proposal density $q(\theta,\theta')$ for the MH update, a number of iterations $S$, a PMF $P_l = \{p_l\}_{l\in \{l_{\text{min}} + 1,\cdots\}}$, and initial values $(\theta_0, \{\bm{Y}_{n,0}^i ,\, V_0^i\}_{i=1}^N)$ such that $\sum_{i=1}^N V_0^i >0$.
\item \textbf{Part 1}: 

\smallskip
\textbf{Initialize:} Set $\widehat{K}=0$ and $D_{\widehat{K}}=0$.

\smallskip
For $k=1,\cdots,S$, do:
\begin{enumerate}[(i)]
\item Sample $\widetilde{\theta}_k$ from $q(\theta_{k-1},\cdot)$.
\item Run \autoref{alg:pf} with input $l=l_{\text{min}}$, $N$, $Y_0$, $n$, $\widetilde{\theta}_k$ and potentials $\{G_p^{(\widetilde{\theta}_k)}\}_{p=1}^n$ to return the output $\{\bm{\widetilde{Y}}_{n,k}^{l,i}, \widetilde{V}_k^i \}_{i=1}^N$.
\item Set
$$
\alpha= \min\left\{1, \frac{q(\widetilde{\theta}_k,\theta_{k-1}) ~\overline{\Pi}(\widetilde{\theta}_k)~(\sum_{i=1}^N \widetilde{V}_k^i +\epsilon) } {q(\theta_{k-1},\widetilde{\theta}_k) ~\overline{\Pi}(\theta_{k-1}) ~ (\sum_{i=1}^N V_{k-1}^i +\epsilon)} \right\}.
$$
Then accept with probability $\alpha$ and set $(\theta_k, \{\bm{Y}_{n,k}^i ,\, V_k^i\}_{i=1}^N) = (\widetilde{\theta}_k, \{\bm{\widetilde{Y}}_{n,k}^i ,\, \widetilde{V}_k^i\}_{i=1}^N)$, $\widehat{K}\leftarrow \widehat{K} + 1$, $D_{\widehat{K}}=1$ and $(\widehat{\theta}_{\widehat{K}}, \{\widehat{\bm{Y}}_{n,\widehat{K}}^i ,\, \widehat{V}_{\widehat{K}}^i\}_{i=1}^N)=(\theta_k, \{\bm{Y}_{n,k}^i ,\, V_k^i\}_{i=1}^N)$. Otherwise set $(\theta_k, \{\bm{Y}_{n,k}^i ,\, V_k^i\}_{i=1}^N) =(\theta_{k-1}, \{\bm{Y}_{n,k-1}^i ,\, V_{k-1}^i\}_{i=1}^N)$ and $D_{\widehat{K}}\leftarrow D_{\widehat{K}}+1$.
\end{enumerate}

\item \textbf{Part 2}: Given $(\widehat{\theta}_k, \{\widehat{\bm{Y}}_{n,k}^i ,\, \widehat{V}_{k}^i\}_{i=1}^N)_{k=1}^{\widehat{K}}$ and $\{D_k\}_{k=1}^{\widehat{K}}$ from Part 1. For $k=1,\cdots,\widehat{K}$, do:
\begin{enumerate}[(i)]
\item For $i=1,\cdots,N$, set $\bm{Y}_{n,k}^{l_{\text{min}},i}=\widehat{\bm{Y}}_{n,k}^i$ and $W_{k,l_{\text{min}}}^i=D_k~\widehat{V}_k^i/(\sum_{i=1}^N \widehat{V}_k^i+\epsilon)$.
\item Sample $l_k$ from $P_l$.
\item Run \autoref{alg:ub_level_diff} with input $l_k$, $N$, $Y_0$, $n$, $\widehat{\theta}_k$, $\lambda_{l_k}$, $F_0^{l_k}$, $\delta_{l_k-1}$, $F_0^{l_k-1}$, $\mu^{l_k}(dx)$ and potential functions $\{G_p^{(\widehat{\theta}_k)},\widecheck{G}_p^{(\widehat{\theta}_k)}\}_{p=1}^n$, to return $\{\bm{Y}_{n,k}^{l_k,i} ,\, \overline{V}_k^{l_k,i}\}_{i=1}^{2N_l}$. For $i=1,\cdots,N$, set 
$$
W_{k,l_k}^i = \frac{\overline{V}_k^{l_k,i}}{p_{l_k}(\sum_{i=1}^N \widehat{V}_k^i+\epsilon)}.
$$
\end{enumerate}
\item \textbf{Output:} Return 

\begin{align}
\label{eq:ub_mean}
\Pi_{ub}(\varphi) := \frac{\sum_{k=1}^{\widehat{K}} \left\{  \sum_{i=1}^N W_{k,l_{\text{min}}}^i ~\varphi(\widehat{\theta}_k,\bm{Y}_{n,k}^{l_{\text{min}},i})  + \sum_{i=1}^{2N} W_{k,l_k}^i~ \varphi(\widehat{\theta}_k,\bm{Y}_{n,k}^{l_k,i})  \right\} }{\sum_{k=1}^{\widehat{K}} \left\{ \sum_{i=1}^N W_{k,l_{\text{min}}}^i  + \sum_{i=1}^{2N} W_{k,l_k}^i \right\} }.
\end{align}

\end{enumerate}
\vspace{-0.2cm}
\hrulefill
\end{center}

\begin{rem}
\begin{enumerate}[(i)]
\item In \autoref{alg:ub_infer}, it is important to note that PMMH is only implemented for the coarsest level, which makes it relatively cheap. In addition, the IS corrections in Part 2 may be calculated in parallel since \autoref{alg:ub_level_diff} is run independently for each $\widehat{\theta_k}$. 

\item As suggested in \cite{CFJLV21}, based on \eqref{eq:bound_on_pf_level_diff} one may take $N_{l_{\text{min}}} = \overline{N}$, $N_l = \widetilde{N}$, with $\widetilde{N} < \overline{N}$, and $p_l\propto 2^{-3l/2}$ for $l>l_{\text{min}}$.

\item Note that the estimator defined in \eqref{eq:ub_mean} is equivalent to the more expensive unbiased estimator
\begin{align}
\label{eq:ub_mean1}
\widehat{\Pi_{ub}(\varphi)} :=\frac{\sum_{k=1}^S \left\{  \sum_{i=1}^N W_{k,l_{\text{min}}}^i ~\varphi(\theta_k,\bm{Y}_{n,k}^{l_{\text{min}},i})  + \sum_{i=1}^{2N} W_{k,l_k}^i~ \varphi(\theta_k,\bm{Y}_{n,k}^{l_k,i})  \right\} }{\sum_{k=1}^S \left\{ \sum_{i=1}^N W_{k,l_{\text{min}}}^i  + \sum_{i=1}^{2N} W_{k,l_k}^i \right\} },
\end{align}
where $(\{W_{k,l_{\text{min}}}^i ,\bm{Y}_{n,k}^{l_{\text{min}},i}\}_{i=1}^N)_{k=1}^S$ and $(\{W_{k,l_k}^i, \bm{Y}_{n,k}^{l_k,i}\}_{i=1}^{2N})_{k=1}^S$ are the output of running Part 2 in \autoref{alg:ub_infer} conditional upon $(\theta_k, \{\bm{Y}_{n,k}^i ,\, V_k^i\}_{i=1}^N)_{k=1}^S$ from Part 1.
\end{enumerate}
\end{rem}

\section{Main Theoretical Results}
In this section we present our theoretical results for the unbiased estimator given in \eqref{eq:ub_mean1} which also apply to the estimator $\Pi_{ub}(\varphi)$ in \eqref{eq:ub_mean}. Let 
$$
\zeta_k(\varphi) :=\sum_{i=1}^N  V_k^i~ \varphi(\theta_k,\bm{Y}_{n,k}^{l_{\text{min}},i}) + p_{l_k}^{-1} \sum_{i=1}^{2N} \overline{V}_k^{l_k,i}~ \varphi(\theta_k,\bm{Y}_{n,k}^{l_k,i})
$$
and
\begin{align*}
\Xi_k(\varphi)&:=\sum_{i=1}^N W_{k,l_{\text{min}}}^i ~\varphi(\theta_k,\bm{Y}_{n,k}^{l_{\text{min}},i})  + \sum_{i=1}^{2N} W_{k,l_k}^i~ \varphi(\theta_k,\bm{Y}_{n,k}^{l_k,i}),
\\
&=\Big(\sum_{j=1}^N V_k^j+\epsilon\Big)^{-1} ~\zeta_k(\varphi).
\end{align*}
Then we can write the estimator in \eqref{eq:ub_mean1} as
$$
\widehat{\Pi_{ub}(\varphi)} = \frac{\sum_{k=1}^S \Xi_k(\varphi)}{\sum_{k=1}^S \Xi_k(1)}=\frac{\sum_{k=1}^S \zeta_k(\varphi)}{\sum_{k=1}^S \zeta_k(1)} \,\,( \text{if }\epsilon>0) .
$$
\begin{lem}
Let $(n,\varphi) \in \mathbb{N} \times \left(\textrm{Lip}_{\|\cdot\|_2}(\mathbb{R}^{nd}) \cap \mathcal{B}_b(\mathbb{R}^{nd})\right)$. Assume \autoref{ass:Gn_and_Mn} holds and the strong error rate $\beta \leq 3$. Then there exist choices of positive probability mass function $P_l=\{p_l\}_{l\in \{l_{\text{min}+1,\cdots}\}}$ with $p_l>0$, and $\{N_l\}_{l\geq l_{\text{min}+1}}$ an increasing sequence of integers with $\lim_{l\to \infty} N_l=\infty$ such that for any $\theta \in \Theta$
\begin{align}
\label{eq:s_varphi}
s_{\varphi}(\theta) := \sum_{l\geq l_{\text{min}}+1}\frac{\mathbb{E}\left[\left( \sum_{i=1}^{2N_l} V^i \, \varphi(\bm{Y}_n^{l,i})\right)^2\right]}{p_l} < \infty,
\end{align}
where $\{\bm{Y}_n^{l,i} ,\, V^{l,i}\}_{i=1}^{2N_l}$ is the output of \autoref{alg:ub_level_diff}.
\end{lem}
\begin{proof}
Fix $\theta \in \Theta$. Using the inequality \eqref{eq:bound_on_pf_level_diff}, the choice of $p_l=(\Delta_l^{\beta})^\rho$ for some $\rho\in (0,1)$ and $N_l=\Delta_l^{-1}$, we have
\begin{align*}
s_{\varphi}(\theta) &\leq C\sum_{l\geq l_{\text{min}}+1}\frac{1}{p_l} \left( \frac{\Delta_l^{2 \wedge \beta}}{N_l} +  \Delta_l^{\beta} \right) \leq C\sum_{l\geq l_{\text{min}}+1} \Delta_l^{\beta(1-\rho)}\left(1+\Delta_l^{2\wedge\beta-\beta+1}\right) < \infty.
\end{align*}
\end{proof}
\begin{lem}
\label{lem:ub1}
Let $(n,\varphi) \in \mathbb{N} \times \left(\textrm{Lip}_{\|\cdot\|_2}(\Theta\times\mathbb{R}^{nd}) \cap \mathcal{B}_b(\Theta\times\mathbb{R}^{nd})\right)$. Assume \autoref{ass:Gn_and_Mn} holds and the strong error rate $\beta \leq 3$. Given the output of Part 1 and Part 2 of \autoref{alg:ub_infer} and that for each $\theta\in \Theta$, $s_\varphi(\theta) < \infty$, then we have 
\begin{align}
&\mathbb{E}\left[\sum_{i=1}^N  V_k^i~ \varphi(\theta_k,\bm{Y}_{n,k}^{l_{\text{min}},i}) + p_{l_k}^{-1} \sum_{i=1}^{2N} \overline{V}_k^{l_k,i}~ \varphi(\theta_k,\bm{Y}_{n,k}^{l_k,i})~\Big|~\theta_k,(V_k^i,\overline{V}_k^{l_k,i},\bm{Y}_{n,k}^{l_k,i})_{i=1}^N\right]\nonumber\\
&=\gamma_n^{(\theta)}(\varphi).
\end{align}
\end{lem}
\begin{proof}
Since $s_\varphi(\theta)<\infty$ for each $\theta\in \Theta$, Theorem 1 in \cite{RG15} applies and one has 
\begin{align*}
\mathbb{E}\left[p_{l_k}^{-1} \sum_{i=1}^{2N} \overline{V}_k^{l_k,i}~ \varphi(\theta_k,\bm{Y}_{n,k}^{l_k,i})~\Big|~\theta_k,(V_k^i,\overline{V}_k^{l_k,i},\bm{Y}_{n,k}^{l_k,i})_{i=1}^N\right] = \gamma_n^{(\theta)}(\varphi) - \gamma_n^{(\theta,l_{\text{min}})}(\varphi).
\end{align*}
We also know that 
$$
\mathbb{E}\left[\sum_{i=1}^N  V_k^i~ \varphi(\theta_k,\bm{Y}_{n,k}^{l_{\text{min}},i})~\Big|~\theta_k,(V_k^i,\overline{V}_k^{l_k,i},\bm{Y}_{n,k}^{l_k,i})_{i=1}^N \right] = \gamma_n^{(\theta,l_{\text{min}})}(\varphi).
$$
This concludes the proof.
\end{proof}
\begin{theorem}[Consistency]
Let $(n,\varphi) \in \mathbb{N} \times \left(\textrm{Lip}_{\|\cdot\|_2}(\Theta\times\mathbb{R}^{nd}) \cap \mathcal{B}_b(\Theta\times\mathbb{R}^{nd})\right)$. Assume \autoref{ass:Gn_and_Mn} holds and that $\beta \leq 3$. Assume that the Markov chain $\{\theta_k,(V_k^i,\bm{Y}_{n,k}^{l_{\text{min}},i})_{i=1}^N\}_{k=1}^S$ (the output of Part 1 of \autoref{alg:ub_infer}) is $\psi$-irreducible and that for each $\theta \in \Theta$
$$
\int \overline{\Pi}(d\theta) \left( \sqrt{s_1(\theta)}+\sqrt{s_\varphi(\theta)}\right) < \infty.
$$
Then,
$$
\widehat{\Pi_{ub}(\varphi)}=\frac{\sum_{k=1}^S \Xi_k(\varphi)}{\sum_{k=1}^S \Xi_k(1)} \xrightarrow{S\to\infty} \Pi(\varphi), \qquad (a.s.).
$$
\end{theorem}
\begin{proof}
Let $\theta_k,(V_k^i,\bm{Y}_{n,k}^{l_{\text{min}},i})_{i=1}^N$, $k=1,\cdots, S$, be the output of Part 1 of \autoref{alg:ub_infer}. Let $f \in \textrm{Lip}_{\|\cdot\|_2}(\Theta\times\mathbb{R}^{nd}) \cap \mathcal{B}_b(\Theta\times\mathbb{R}^{nd})$ and define
\begin{align*}
\mu_f\left(\theta,(V^i,\bm{Y}_n^{l_{\text{min}},i})_{i=1}^N\right) &:= \mathbb{E}\left[\Xi_k(f)\Big|\left(\theta_k,(V_k^i,\bm{Y}_{n,k}^{l_{\text{min}},i})_{i=1}^N\right)=\left(\theta,(V^i,\bm{Y}_n^{l_{\text{min}},i})_{i=1}^N\right)\right]\\
&=\Big(\sum_{i=1}^N V^i+\epsilon\Big)^{-1} ~\gamma_n^{(\theta)}(f) \quad(\text{by \autoref{lem:ub1}}),
\end{align*}
and
\begin{align*}
m^{(1)}_f &\left(\theta,(V^i,\bm{Y}_n^{l_{\text{min}},i})_{i=1}^N\right)\\
 &:= \mathbb{E}\left[|\Xi_k(f)|\Big|\left(\theta_k,(V_k^i,\bm{Y}_{n,k}^{l_{\text{min}},i})_{i=1}^N\right)=\left(\theta,(V^i,\bm{Y}_n^{l_{\text{min}},i})_{i=1}^N\right)\right]\\
&\leq \Big(\sum_{i=1}^N V^i+\epsilon\Big)^{-1}\left[\sum_{i=1}^N V^i\big|f(\theta,\bm{Y}_n^{l_{\text{min}},i})\big|+\sqrt{s_f(\theta)}\right].
\end{align*}
Define the probability
$$
\widetilde{\Pi}\left(d\theta,d(V^i,\bm{Y}_n^{l_{\text{min}},i})_{i=1}^N \right) =\frac{1}{Z} \Big(\sum_{i=1}^N V^i+\epsilon\Big) ~d(V^i)_{i=1}^N ~\overline{\Pi}(d\theta) ~\gamma_n^{(\theta,l_{\text{min}})}(d(\bm{Y}_n^{l_{\text{min}},i})_{i=1}^N),
$$
where $Z$ is the normalization constant, then the chain $\{\theta_k,(V_k^i,\bm{Y}_{n,k}^{l_{\text{min}},i})_{i=1}^N\}_{k=1}^S$, is reversible w.r.t. $\widetilde{\Pi}$. Since the chain is also $\psi$-irreducible, it is Harris recurrent. Moreover,
\begin{align*}
\widetilde{\Pi}(m^{(1)}_\varphi) \leq C \gamma_n^{(\theta)}(\varphi) ~ \int \overline{\Pi}(d\theta) ~\sqrt{s_\varphi(\theta)} <\infty,
\end{align*}
where $C>0$ is some constant and we used the fact that $f$ is bounded, the assumption that $\gamma_n^{(\theta)}$ is finite and $(V_i)_{i=1}^N$ are finite. Similarly, $\widetilde{\Pi}(m^{(1)}_1) < \infty$. Notice also that 
\begin{align*}
\widetilde{\Pi}(\mu_\varphi)=\frac{1}{Z}~\gamma_n^{(\theta,l_{\text{min}})}(\varphi) \int \overline{\Pi}(d\theta) ~\gamma_n^{(\theta)}(d(\bm{Y}_n^{l_{\text{min}},i})_{i=1}^N) = C ~\Pi(\varphi).
\end{align*}
Similarly, $\widetilde{\Pi}(\mu_\varphi)=C$ some finite constant. Thus, the assumptions of \cite[Theorem 1]{VHF20} holds and the result follows.
\end{proof}

\section{Numerical Simulations}
\label{sec:numer}

Consider the model in \ref{eq:levy_driv_sde}-\ref{eq:obs}, where $\{X_t\}_{t\in[0,T_f]}$ is a 1-dimensional L\'{e}vy process characterized by the triplet $(\nu,\Sigma,b)$. We consider a symmetric truncated L\'{e}vy process, that is, $\Sigma=b=0$ which implies that the process considered here has no drift and Brownian motion components. We define the L\'{e}vy measure $\nu$ on $\mathbb{R}\setminus\{0\}$ as
\begin{align*}
\nu(dx) =\left( \mathbb{I}_{[-u,0)}(x) \frac{c}{(-x)^{1+\alpha}} + \mathbb{I}_{(0,u]}(x) \frac{c}{x^{1+\alpha}}  \right) ~dx,
\end{align*}
where $c>0$, $\alpha\in (0,2)$ and $u\geq 1$ is the truncation threshold. This is the same L\'{e}vy measure used in \cite{DH11,AKP19}. Then the measure $\mu^l$, $l\in \mathbb{N}$, is given by
\begin{align*}
\mu^l(dx) =\frac{c}{\lambda_l} \left[\mathbb{I}_{[-u,-\delta_l]}(x) \frac{1}{(-x)^{1+\alpha}} + \mathbb{I}_{[\delta_l,u]}(x) \frac{1}{x^{1+\alpha}}\right]~dx,
\end{align*}
where $\lambda_l$ can be computed analytically as
$$
\lambda_l = \nu(B_{\delta_l}^c) = 2c \int_{\delta_l}^u x^{-1-\alpha}~ dx = \frac{2c}{\alpha}\left(\frac{1}{\delta_l^\alpha} -\frac{1}{u^\alpha} \right).
$$
Recall that $\delta_l$ is chosen such that $\lambda_l = \Delta_l^{-1}$, therefore, we have
$$
\delta_l = \left(\frac{\alpha}{2c\Delta_l} + \frac{1}{u^\alpha}  \right)^{-1/\alpha}.
$$
\begin{figure}[h!]
\centering
\begin{tikzpicture}[scale=1.3]
  
\def\u{1}
\def\a{0.5}
\def\c{0.5}
\def\l{2};
\pgfmathsetmacro{\D}{2^(-\l)}
\pgfmathsetmacro{\L}{1/\D}
\pgfmathsetmacro{\d}{(\a/(2*\c*\D)+(\u)^(-\a))^(-1/\a)}

\begin{axis}[ticks= none, axis x line=middle, axis y line=middle, axis equal]

\addplot[blue, line width=1pt, domain=-\u:-\d]{ (\c/\L) *abs(x)^(-1-\a) };
\addplot[blue, line width=1pt, domain=\d:\u]{ (\c/\L) *abs(x)^(-1-\a)  };
\addplot[dashed, line width=1pt, blue] coordinates{ (-\d, 0)  (-\d, {(\c/\L) *  (\d)^(-3/2)} ) };
\addplot[dashed, line width=1pt, blue] coordinates{ (\d, 0)  (\d, {(\c/\L) *  (\d)^(-3/2)} ) };
\addplot[line width=1pt,blue] coordinates{(-\d,0) (\d,0)};

\addplot[line width=1pt,blue, dashed] coordinates{(-\u,0) (-\u,{(\c/\L) *(\u)^(-1-\a) })};

\addplot[line width=1pt,blue, dashed] coordinates{(\u,0) (\u, {(\c/\L) *(\u)^(-1-\a) })};

\addplot[red, line width=1pt, domain=-\u:-\d]{ ( \c/(\L*\a)) *(abs(x)^(-\a)-(\u)^(-\a) };

\addplot[line width=1pt, red,  domain=-\d:\d] { ( \c/(\L*\a)) *((\d)^(-\a) -(\u)^(-\a)};

\addplot[red, line width=1pt, domain=\d:\u]{ ( \c/(\L*\a)) *(2*(\d)^(-\a) -x^(-\a)-(\u)^(-\a) };


\addplot[black, line width=1pt, domain=0:(\c/(\L*\a)) *((\d)^(-\a) -(\u)^(-\a)]{ -(\L*\a/\c *x + (\u)^(-\a))^(-1/\a) };

\addplot[black, line width=1pt, dashed] coordinates{({(\c/(\L*\a)) *((\d)^(-\a) -(\u)^(-\a)},-\d) ({(\c/(\L*\a)) *((\d)^(-\a) -(\u)^(-\a)}, \d)};

\addplot[black, line width=1pt, domain=(\c/(\L*\a)) *((\d)^(-\a) -(\u)^(-\a):1]{ (2*(\d)^(-\a)-\L*\a/\c *x - (\u)^(-\a))^(-1/\a) };

\node at (axis cs:-\u,-0.15) {\scriptsize $-u$};
\node at (axis cs:-\d-0.1,-0.15) {\scriptsize $-\delta_l$};
\node at (axis cs:\d+0.03,-0.15) {\scriptsize $\delta_l$};
\node at (axis cs:\u,-0.15) {\scriptsize $u$};

\addplot[black, line width=0.5pt] coordinates{(-0.03,1) (0.05,1)};
\node at (axis cs:0.19,1) {\scriptsize 1};

\end{axis}

\end{tikzpicture}
\caption{The blue, red and black curves correspond to the density of $\mu^l$, the CDF and the inverse of the CDF, respectively.}
  \label{fig:mu}
\end{figure}
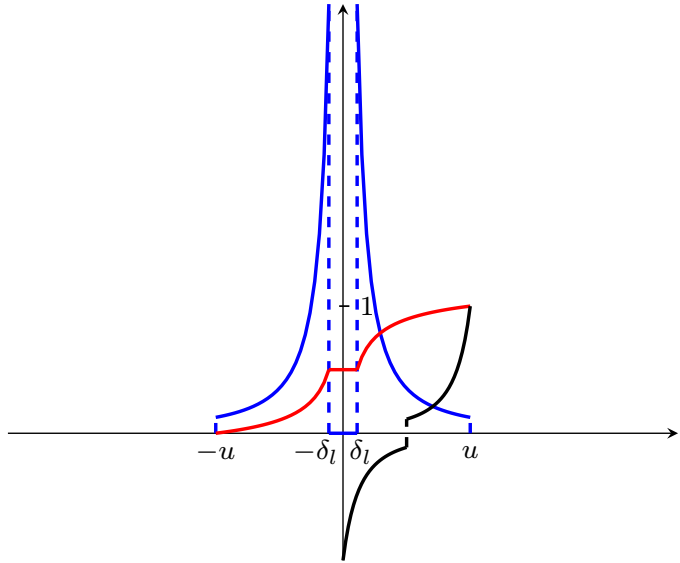

\noindent Notice that $\delta_l < u$, so the definition of $\mu^l$ makes sense. We take $u=1$, $c=0.8$ and $\alpha=0.5$. \autoref{fig:mu} shows the graphs of the density associated with $\mu^l$ for $l=2$, the cumulative distribution function (CDF) and its inverse. Note that due to the symmetric structure of $\nu$, we have $F_0^l = 0$. Since the CDF inverse exists and can be computed analytically, we use the inverse transform sampling method to sample the heights of the L\'{e}vy process jumps from $\mu^l$ in Step 4. of \autoref{alg:disc_levy}. One can compute the strong error rate of the Euler scheme associated with the measure $\nu$. The term $\sigma_{\delta_l}^2$ in \eqref{eq:strong_err_Euler} is computed analytically as
$$
\sigma_{\delta_l}^2 = \int_{B_{\delta_l}} |x|^2 \nu(dx) = 2c \int_0^{\delta_l} x^{1-\alpha} dx = \frac{2c}{2-\alpha} \delta_l^{2-\alpha} = \mathcal{O}(\delta_l^{2-\alpha}),
$$
hence by \eqref{eq:strong_err_Euler},
$$
\mathbb{E} \Big[ \sup_{t\in [0,1]} \left|Y_t - Y_t^l \right|^2 \Big] \leq C \delta_l^{2-\alpha} = C\left(\frac{\alpha}{2c\Delta_l} + 1  \right)^{-(2-\alpha)/\alpha} = \mathcal{O}(\Delta_l^{(2-\alpha)/\alpha}).
$$
Since $\alpha=0.5$, we have the strong error rate $\beta=3$.

The L\'{e}vy-driven SDE considered here has the form
$$
dY_t = \theta ~Y_{t^-} ~ dX_t, \quad Y_0 = y_0,
$$
with $y_0=1$. The observations set is the daily S \& P 500 log-returns from Jan 3, 2012 to May 25, 2013 taken from Yahoo Finance website. In \autoref{fig:mse_cost}, we plot the MSE for \autoref{alg:ub_infer} and PMMH against the computational cost measured in seconds. The MSE is computed by running 52 independent simulations of each method
\[\text{MSE}_{ub} = \frac{1}{52} \sum_{i=1}^{52} [\Pi_{ub}(\varphi)^i - \Pi(\varphi)]^2,\qquad \text{MSE}_{pmmh} = \frac{1}{52} \sum_{i=1}^{52} [\Pi_{pmmh}(\varphi)^i - \Pi(\varphi)]^2,\]
where $\Pi(\varphi)$ is the ground truth. We explain below how the ground truth is estimated. In practice, one has to truncate the values of $l$ in Part 2 of \autoref{alg:ub_infer}, therefore, for each simulation $\Pi_{ub}(\varphi)^i$ in the MSE identity above we set $P_l = (p_l)_{l=l_{\text{min}+1}}^{l_{\text{max}}}$ with $l_{\text{min}} = 1$, $l_{\text{max}} = 12$ and $p_l = 2^{-3l/2}$. The reference $\Pi(\varphi)$ is estimated as the mean of 52 runs of \autoref{alg:ub_infer} with $l_{\text{max}} = 14$, $N=100$ and $S=10^5$. We choose a positive algorithm constant $\epsilon = 10^{-8}$. The number of particles is at most 60 in both algorithms and PMMH has a discretization level of at most 8 (higher levels will result in a much larger computational cost.) We set the function $\varphi(\theta, x) = \theta$ in our simulations. In \autoref{tbl:simul_res},  we compare the expected value of $\theta$, the MSE and ratio of computational costs for both methods. The reference value is 0.76249.

\begin{figure}
\centering
\includegraphics[width =0.7\textwidth]{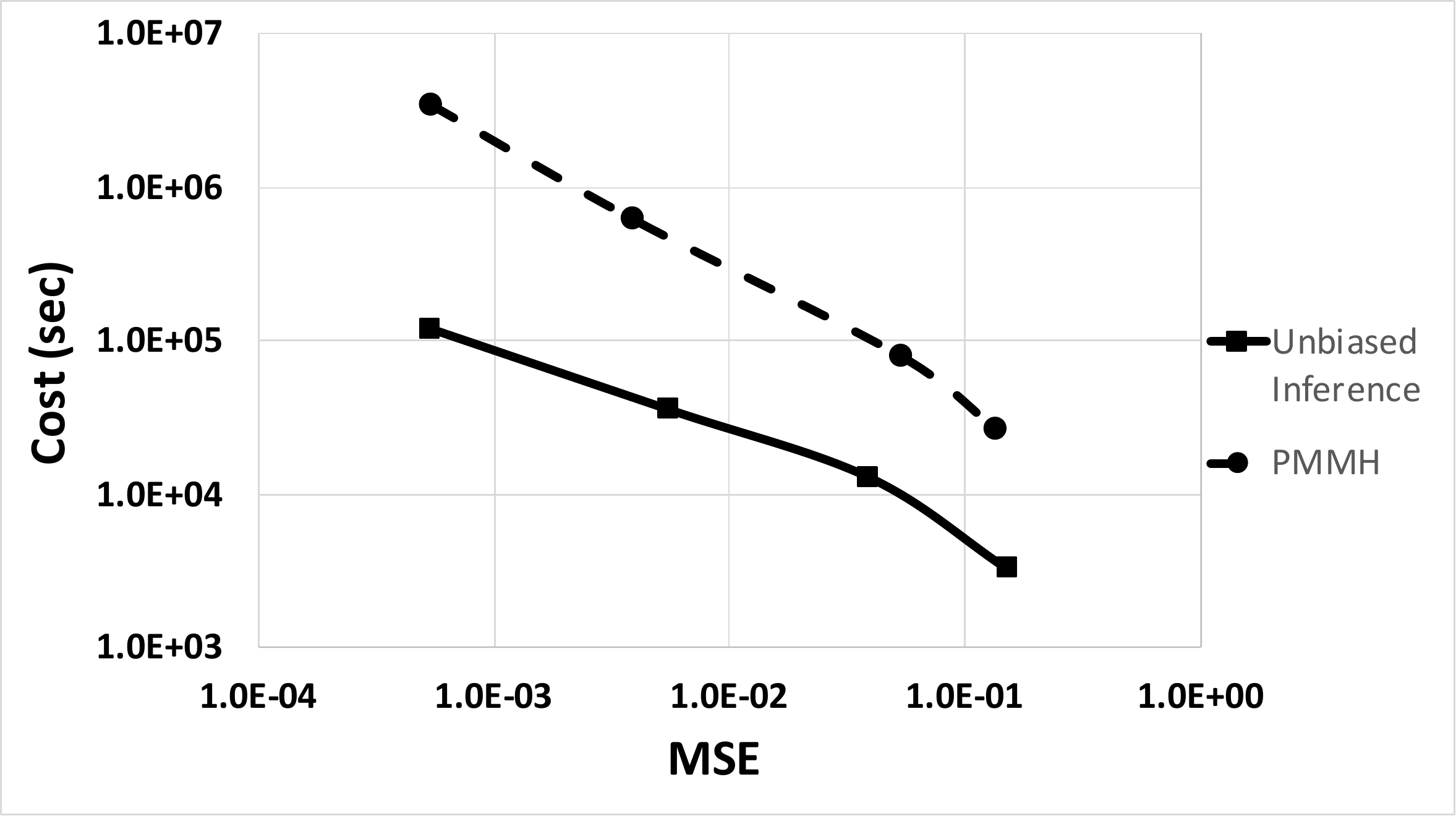}
\caption{The computational cost (in seconds) versus MSE for both the unbiased inference method in \autoref{alg:ub_infer} and PMMH algorithm in \cite{ADH10}.}
\label{fig:mse_cost}
\end{figure}

\begin{table}[h!]
\centering
\begin{tabular}{|l|l|l|l|c|}
\hline
$\theta_{ub}$ & $\text{MSE}_{ub}$     & $\theta_{pmmh}$  & $\text{MSE}_{pmmh}$ & $\text{Cost}_{pmmh}/\text{Cost}_{ub}$\\ \hline
0.78446 & 1.5054 E-01 & 0.78748     & 1.3591 E-01   &  7.98 \\ \hline
0.76129 & 3.9087 E-02 & 0.78309     & 5.4409 E-02 &  6.02 \\ \hline
0.76089 & 5.4825 E-03 & 0.78069    & 3.9315 E-03 &  17.1 \\ \hline
0.76183 & 5.3014 E-04 & 0.77012    & 5.4522 E-04 &  28.5 \\ \hline
\end{tabular}
\caption{First and third columns are the expected values of $\theta$ obtained by \autoref{alg:ub_infer} and PMMH, respectively. The second and fourth columns are the corresponding MSE and the fifth column is the ratio of the computational costs.}
\label{tbl:simul_res}
\end{table}

\end{document}